\documentclass{article}
\usepackage[utf8]{inputenc}

\title{Approximate Polymorphisms}
\author{Gilad Chase\thanks{The Henry and Marylin Taub Faculty of Computer Science, Technion, Israel.} \qquad Yuval Filmus\thanks{The Henry and Marylin Taub Faculty of Computer Science, Technion, Israel. This project has received funding from the European Union's Horizon 2020 research and innovation programme under grant agreement No~802020-ERC-HARMONIC.} \qquad Dor Minzer\thanks{Department of Mathematics, Massachusetts Institute of Technology, Cambridge, USA.} \qquad Elchanan Mossel\thanks{Department of Mathematics and IDSS, Massachusetts Institute of Technology, Cambridge, USA.} \qquad Nitin Saurabh\thanks{Department of Computer Science, Technion, Israel. This project has received funding from the European Union's Horizon 2020 research and innovation programme under grant agreement No~802020-ERC-HARMONIC.}}
\date{\vspace{-5ex}}
\usepackage{amsmath,amsfonts,amssymb,amsthm,enumerate}
\usepackage[pagebackref,colorlinks,linkcolor=blue,citecolor=red]{hyperref}
\usepackage[nameinlink]{cleveref}
\usepackage{crossreftools,color,array,xfrac}
\usepackage{tcolorbox}
\usepackage{fullpage}

\newif\ifbackrefshowonlyfirst
\backrefshowonlyfirstfalse
\makeatletter
\let\BR@direct@old@hyper@natlinkstart\hyper@natlinkstart
\renewcommand*{\hyper@natlinkstart}{\phantomsection\BR@direct@old@hyper@natlinkstart}% note that the anchor will appear after any brackets at the start of the citation, but that's not really a big issue?
\let\BR@direct@oldBR@citex\BR@citex
\renewcommand*{\BR@citex}{\phantomsection\BR@direct@oldBR@citex}%

\long\def\hyper@page@BR@direct@ref#1#2#3{\hyperlink{#3}{#1}}

\ifx\backrefxxx\hyper@page@backref
    \let\backrefxxx\hyper@page@BR@direct@ref
    \ifbackrefshowonlyfirst
    \fi
\else
    \ifbackrefshowonlyfirst
    \fi
\fi

\RequirePackage{etoolbox}
\patchcmd{\Hy@backout}{Doc-Start}{\@currentHref}{}{\errmessage{I can't seem to patch backref}}
\makeatother
\newtheorem{theorem}{Theorem}[section]
\newtheorem{lemma}[theorem]{Lemma}

\newtheorem{open}{Open Question}
\theoremstyle{definition}
\newtheorem{definition}[theorem]{Definition}
\theoremstyle{remark}
\newtheorem{remark}[theorem]{Remark}

\providecommand{\row}{\mathsf{row}}
\providecommand{\col}{\mathsf{col}}
\providecommand{\rr}[1]{\mathcal{R}_{#1}}
\providecommand{\biglor}{\bigvee}
\providecommand{\bigland}{\bigwedge}
\providecommand{\h}{g}

\DeclareMathOperator*{\E}{\mathbb{E}}
\providecommand{\V}[1]{\mathsf{var}(#1)}
\DeclareMathOperator{\NS}{NS}
\DeclareMathOperator{\Inf}{Inf}
\DeclareMathOperator{\Maj}{Maj}
\DeclareMathOperator{\Stab}{Stab}
\providecommand{\cG}{\mathcal{G}}
\providecommand{\cH}{\mathcal{H}}
\providecommand{\normal}{\mathcal{N}}
\providecommand{\normalg}[1]{\normal_{#1}}
\providecommand{\stdnormal}{\normal(0,1)}
\DeclareMathOperator{\clip}{\mathsf{clip}}
\DeclareMathOperator{\sgn}{\mathsf{sign}}
\providecommand{\half}{\sfrac12}

\newcommand\eps{\varepsilon}
\renewcommand\epsilon{\eps}
\renewcommand\geq{\geqslant}
\renewcommand\leq{\leqslant}
\renewcommand\ge{\geqslant}

\begin{document}

\maketitle

\begin{abstract}
For a function $g\colon\{0,1\}^m\to\{0,1\}$, a function $f\colon \{0,1\}^n\to\{0,1\}$ is called a $g$-polymorphism if their actions commute: $f(g(\row_1(Z)),\ldots,g(\row_n(Z))) = g(f(\col_1(Z)),\ldots,f(\col_m(Z)))$ 
for all $Z\in\{0,1\}^{n\times m}$. The function $f$ is called an approximate $g$-polymorphism if this equality holds with probability close to $1$, when 
$Z$ is sampled uniformly. A pair of functions $f_0,f_1\colon \{0,1\}^n \to \{0,1\}$ are called a skew $g$-polymorphism if $f_0(g(\row_1(Z)),\ldots,g(\row_n(Z))) = g(f_1(\col_1(Z)),\ldots,f_1(\col_m(Z)))$ 
for all $Z\in\{0,1\}^{n\times m}$.

We study the structure of exact polymorphisms as well as approximate polymorphisms. 
Our results include:
\begin{enumerate}
    \item We prove that an approximate polymorphism $f$ must be 
    close to an exact \emph{skew} polymorphism;
    \item We give a characterization of exact skew polymorphisms, showing that besides
    trivial cases, only the functions $g = \mathsf{AND}, \mathsf{XOR}, \mathsf{OR}, \mathsf{NAND}, \mathsf{NOR}, \mathsf{NXOR}$ admit non-trivial exact skew polymorphisms.
\end{enumerate}
We also study the approximate polymorphism problem in the list-decoding regime (i.e., when the probability equality holds is not close to $1$, but is bounded away from some value).
We show that if $f(x \land y) = f(x) \land f(y)$ with probability larger than $s_\land \approx 0.815$ then $f$ correlates with some low-degree character, and $s_\land$ is the optimal threshold for this property.

Our result generalize the classical linearity testing result of Blum, Luby and Rubinfeld, that in this language showed that the approximate polymorphisms of 
$g = \mathsf{XOR}$ are close to XOR's, as well as a recent result of Filmus, Lifshitz, Minzer 
and Mossel, showing that the approximate polymorphisms of AND can only be close
to AND functions.
\end{abstract}

\section{Introduction}
\label{sec:introduction}
Let $m\in\mathbb{N}$ be thought of as a constant, $n\in\mathbb{N}$ be thought of as large, and let $g\colon\{0,1\}^m\to\{0,1\}$ be any function. We say that $f\colon\{0,1\}^n\to\{0,1\}$ is a \emph{polymorphism} of $g$ if their operations 
commute. More precisely, defining the functions 
$f\circ g^{n}, g\circ f^{m} \colon \{0,1\}^{n\times m}\to\{0,1\}$ 
as
\[
 (f\circ g^{n})(Z) = f(g(\row_1(Z)),\ldots,g(\row_n(Z))),
 \qquad
 (g\circ f^{m})(Z) = g(f(\col_1(Z)),\ldots,f(\col_m(Z))),
\]
we say that $f$ is a polymorphism of $g$ if $f\circ g^n = g\circ f^m$.
See \Cref{fig:polymorphism} for an illustration.

\begin{figure}[!ht]
\[
\begin{array}{|c|c|c|c|c|}
\cline{1-3} \cline{5-5}
Z_{11} & \cdots & Z_{1m} & \stackrel g\to & g(\row_1(Z))
\\ \cline{1-3} \cline{5-5}
\vdots & \ddots & \vdots & \raisebox{2pt}{$\stackrel g\to$} & \vdots 
\\ \cline{1-3} \cline{5-5}
Z_{n1} & \cdots & Z_{nm} & \stackrel g\to & g(\row_n(Z))
\\ \cline{1-3} \cline{5-5}
\multicolumn{1}{c}{\downarrow \raisebox{1pt}{$\scriptstyle f$}} &
\multicolumn{1}{c}{\downarrow \raisebox{1pt}{$\scriptstyle f$}} &
\multicolumn{1}{c}{\downarrow \raisebox{1pt}{$\scriptstyle f$}} &
\multicolumn{1}{c}{} &
\multicolumn{1}{c}{\downarrow \raisebox{1pt}{$\scriptstyle f$}}
\\ \cline{1-3} \cline{5-5}
f(\col_1(Z)) & \cdots & f(\col_m(Z)) & \stackrel g\to & \ast
\\ \cline{1-3} \cline{5-5}
\end{array}
\]
    \caption{$f$ is a polymorphism of $g$, in symbols $f \circ g^n = g \circ f^m$, if applying $g$ to the bottom row produces identical results to applying $f$ to the rightmost column.}
    \label{fig:polymorphism}
\end{figure}
%\skipi

More generally, for a parameter $\delta>0$, we say that $f$ is a \emph{$\delta$-approximate 
polymorphism} if 
\[
\Pr_{Z}[(f \circ g^{n})(Z) \neq (g \circ f^{m})(Z)]\leq \delta;
\]
here and throughout, the distribution over $Z$ is uniform over $\{0,1\}^{n\times m}$. We note that for any function $g$, one always has dictatorship functions as polymorphisms. Namely, for each $j\in[n]$, it is
easily seen that the function $f(x) = x_j$ is a polymorphism of $g$ as 
\[
(f \circ g^{n})(Z)
=g(\row_j(Z))
=(g\circ f^{m})(Z).
\]
Dictatorship polymorphisms will thus be referred to as \emph{trivial polymorphisms} of $g$. If $g$ possesses a mild structural property, then there are additional 
trivial polymorphisms: when $g$ is odd, anti-dictatorships also form 
polymorphisms; if $g(b,\ldots,b) = b$, then the constant function $f(x) = b$ also forms a polymorphism.
What can be said about the structure of functions $g$ that have non-trivial polymorphisms? More generally, what can be said about functions $g$ that have approximate polymorphisms that are far from being trivial? Furthermore, can we classify the structure of the approximate polymorphisms in these cases?

The problem of studying the structure of polymorphisms as well as approximate polymorphisms has appeared in several different contexts throughout theoretical computer science:
\begin{enumerate}
    \item 
    \textbf{Universal algebra and the complexity of constraint satisfaction problems.} In this context, the function $g$ is allowed to be a predicate rather than a function, and a polymorphism is a function $f$ that takes satisfying assignments to $g$, ordered as rows in the matrix $Z\in\{0,1\}^n$, and produces a satisfying assignment for $g$ in the form $f^{m}(Z)$ (i.e., somewhat of a one-sided version of the above equation). In this context, the existence of non-trivial polymorphisms is strongly linked 
    to the complexity of the constraint satisfaction problem corresponding to the 
    predicate $g$ (see for example~\cite{CSPpoly}).
    
    \item \textbf{Property testing.} Perhaps the most basic problem in property testing, the linearity testing problem~\cite{BLR,Bellare}, can be cast in the language of approximate polymorphisms. Here, one takes $m=2$ and the function $g(x,y) = x\oplus y$, in which case a function $f$ is an $\delta$ approximate polymorphism 
    if $f(x\oplus y) = f(x)\oplus f(y)$ with probability $\geq 1-\delta$, where 
    $x$ and $y$ are sampled uniformly and independently from $\{0,1\}^n$. This question, as well as its $\half + \delta$ list-decoding variant have been 
    well studied and are useful in the study of PCP's~\cite{Hastad}.
    
    \item \textbf{Social choice theory.} In this context, one thinks of the functions 
    $f,g$ as voting rules, and then the above functions $f\circ g^n$, 
    $g\circ f^m$ can be thought of as two ways of aggregating these 
    voting rules in order to reach a final outcome. %In this context,
    Due to this interpretation,
    it makes sense
    to also consider the ``cross'' version of the problem, wherein we have multiple functions $f$, say $f_0,\ldots,f_m$, and we replace the above equation by
    \[
    (f_0\circ g^{n})(Z) = g(f_1(\col_1(Z)),\ldots,f_m(\col_m(Z))).
    \]
    The interpretation here is that there are $n$ voters that cast their yes/ no opinion on each one of $m$ topics; notationally, the vector 
    $\row_i(Z)$ represents the opinions of voter $i$. The goal is to aggregate
    these opinions about the topic to reach a final conclusion, and naturally
    this can be done in one of two ways: first, one may aggregate the opinion 
    of each voter, and then aggregate the final conclusion of each voter. This
    way of aggregation is represented by $(f_0\circ g^{n})(Z)$. Another way
    to aggregate these opinions is to first reach a final conclusion regarding
    each topic, which is $f_j(\col_j(Z))$ in the above notation, and then aggregate those; this is represented by the function $g(f_1(\col_1(Z)),\ldots,f_m(\col_m(Z)))$. Thus, in this interpretation, the 
    question asks for which aggregation rules $g$ and $f_i$ does it hold that 
    the two natural ways of aggregating the votes are essentially equivalent.
    
    The case where $g$ is an AND function is a prominent example that has been 
    studied in this context. In particular, the fact 
    that $f = {\sf Majority}$ does not yield equivalent rules is known as the 
    \emph{Doctrinal paradox}, which raised the question of what are all $f$'s in
    this case that yield equivalent rules. This problem has been addressed by Nehama~\cite{Nehama} in the context of social choice theory and by Parnas, Ron and Samorodnitsky~\cite{PRS02} from the property testing point of view, and both works establishing partial results. A recent work~\cite{FLMM2020} has 
    improved these results, showing that in this case 
    the approximate polymorphisms of $g$ can only be functions that are close to AND functions.
\end{enumerate}

With this in mind, it makes sense to ask what is the most general result one can prove when $g$ is a general function on constantly many coordinates. Indeed, answering this question is the main goal of this paper:
\begin{tcolorbox}
Determine all pairs $f,g$ which are \emph{approximate polymorphisms}:
\[
 \Pr[f \circ g^n = g \circ f^m] \geq 1-\delta.
\]
\end{tcolorbox}

\subsection{The structure of exact polymorphisms}
The exact polymorphisms variant of this problem, i.e.\ the case that 
$\delta = 0$, has been previously studied by Dokow and Holzman~\cite{DH09}. 
They manage to give the following tight classification of all possible pairs $f,g$ 
in which $f$ is a polymorphism of $g$:
\begin{enumerate}
\item One of $f,g$ is constant, a dictator ($x_i$), or an anti-dictator ($\lnot x_i$).
\item $f,g$ are XORs or their negation.
\item $f,g$ are ANDs.
\item $f,g$ are ORs.
\end{enumerate}
Stated otherwise, the only $g$'s that have non-trivial polymorphisms are 
AND's, OR's, XOR's and NXOR's. It is interesting to note that in each one
of these cases, the answer to the approximate polymorphisms problem has
already been resolved; the case $g$ is an XOR or an NXOR is linearity testing~\cite{BLR,Bellare}, and it is well-known that $f$ must be close to an XOR or its negation. When $g$ is an AND, it was shown in~\cite{FLMM2020} that $f$ is close to zero or to an AND, and the case where $g$ is an OR is similar.

Thus, it would be natural to guess that the only $g$'s that have non-trivial approximate polymorphisms would be exactly the $g$'s found by Dokow and Holzman.
Furthermore, we would expect that if $g$ is not an XOR, NXOR, AND, or OR, and 
$f$ is an approximate polymorphism, then $f$ must be trivial, i.e. close to a constant, a dictator, or an anti-dictator. Here and throughout, closeness is 
measured with respect to the Hamming distance over the uniform measure on
$\{0,1\}^n$. 

\subsection{Main results}
\subsubsection{Approximate polymorphisms}
In this language, our main result reads:
\begin{theorem} \label{thm:main-intro}
Fix $g\colon \{0,1\}^m \to \{0,1\}$. For every $\epsilon > 0$ there exists $\delta > 0$ (depending on both $g$ and $\epsilon$) such that if $f\colon \{0,1\}^n \to \{0,1\}$ satisfies
\[
 \Pr_Z[(f \circ g^{n})(Z) = (g \circ f^{m})(Z)] \geq 1-\delta,
\]
then either:
\begin{enumerate}
    \item $f$ is $\eps$-close to a constant, dictator, anti-dictator or an exact polymorphism of $g$;
    \item $g$ is either an NOR or an NAND, and $f$ is $\eps$-close to an OR or an AND (respectively).
\end{enumerate}
\end{theorem}
Naively, one may have hoped that the first item in the theorem must always hold, 
however as we explain next, it is necessary to include the second item as well. Suppose that $g$ is unbalanced, so that $p = \E[g]$ is at least $2^{-m}$-far from $1/2$. Suppose we have functions $f_0$ and $f_1$ satisfying 
$f_0 \circ g^{n} = g \circ f_1^{m}$. Given such $f_0, f_1$, we may construct a function $f$ agreeing with $f_1$ around the middle slice and with $f_0$ around the $pn$-slice, and have that $f \circ g^{n} \approx g \circ f^m$. All new solutions in \Cref{thm:main-intro} arise from such \emph{skew polymorphisms}.

\smallskip

Dokow and Holzman in fact solved the more general cross version of the problem defined above. Namely, they managed to classify all solutions to the equation
$ f_0 \circ g^{n} = g \circ (f_1,\ldots,f_m)$. Our main result extends 
to this setting as well, and we prove an analog of \Cref{thm:main-intro} for it 
as well (see~\Cref{thm:main-multi} for a precise statement).

We also provide an alternative proof of the classification of Dokow and Holzman~\cite{DH09}, using Boolean function analysis. To illustrate the merits of this proof technique, we classify all solutions of the slightly more general equation
\[
 f_0 \circ g^{n} = h \circ (f_1,\ldots,f_m).
\]

\subsubsection{The list decoding regime} 
As discussed earlier, the linearity testing problem, which constitutes 
one example of the approximate polymorphisms problem, can be studied in
several different regimes:
\begin{enumerate}
    \item Exact regime: If $\Pr[f(x \oplus y) = f(x) \oplus f(y)] = 1$ then $f$ is an XOR.
    \item Approximate regime: If $\Pr[f(x \oplus y) = f(x) \oplus f(y)] \geq 1-\delta$ then $f$ is $O(\delta)$-close to an XOR.
    \item List decoding regime: If $\Pr[f(x \oplus y) = f(x) \oplus f(y)] \geq 1/2 + \delta$ then $f$ is $\Omega(\delta)$-correlated with some XOR.
\end{enumerate}
In this language, Dokow and Holzman extended the exact regime to arbitrary functions $g$ (in linearity testing, $g(x,y) = x \oplus y$), and \Cref{thm:main-intro} extends the approximate regime to arbitrary functions $g$. Our second main result extends the list decoding regime to arbitrary functions $g$.

\begin{theorem} \label{thm:list-decoding-intro}
Fix $g\colon \{0,1\}^m \to \{0,1\}$ which is not a constant, a dictator, or an anti-dictator. There exists a constant $s_g < 1$ such that the following holds:

\begin{enumerate}
    \item For every $\epsilon > 0$ there exists $\delta > 0$ such that if
\[
 \Pr_Z[(f \circ g^{n})(Z) = (g \circ f^{m})(Z)] \geq s_g + \delta
\]
then $f$ is $\epsilon$-correlated with some XOR.
\item For every $\delta>0$, there exists large enough $n$ and a function 
$f\colon\{0,1\}^n\to\{0,1\}$ such that the correlation of $f$ with any XOR is
at most $\delta$, and 
\[
\Pr_Z[(f \circ g^{n})(Z) = (g \circ f^{m})(Z)] \leq s_g - \delta.
\]
\end{enumerate}
\end{theorem}

Computing the value of $s_g$ may be a challenging task in general. 
In the case that $g$ is an XOR, one has $s_g = 1/2$. When $g(x,y) = x\land y$, 
one may have expected $s_g$ to be equal to $3/4$ (as was conjectured by~\cite{FLMM2020}). It turns out that it is actually higher, about $0.814975$.

\subsubsection{Comparison to previous work}

Our main theorems generalize classical work on linearity testing~\cite{BLR,Bellare}, which is the special case of $g = \mathsf{XOR}$.

This work was prompted by recent work~\cite{FLMM2020} which proved \Cref{thm:main-intro} in the special case of $g = \mathsf{AND}$.
\Cref{thm:main-intro} and \Cref{thm:list-decoding-intro} answer two of the three open questions posed in~\cite{FLMM2020}.

Several other results in the literature can be seen as analogs of \Cref{thm:main-intro} in the more general case of predicates:
\begin{enumerate}
    \item \textbf{Arrow's theorem.} The $\mathsf{NAE_3}$ predicate is a predicate on triples of bits which holds whenever the three bits are not all equal.
    Arrow's theorem~\cite{Arrow} for three candidates states, in our language, that the only polymorphisms of $\mathsf{NAE_3}$ satisfying $f(0,\ldots,0) = 0$ and $f(1,\ldots,1) = 1$ are dictators. Wilson~\cite{Wilson} improved this, showing that the only polymorphisms of $\mathsf{NAE_3}$ are dictators and anti-dictators.
    \\
    Mossel~\cite{Mossel2012}, improving on earlier work of Kalai~\cite{Kalai}, showed that approximate polymorphisms of $\mathsf{NAE_3}$ are close to dictators or to anti-dictators.
    Mossel's result is more general, allowing for more than three candidates.\footnote{Mossel's result is even more general, classifying also exact and approximate \emph{multi-}polymorphisms, in which different $f$'s are allowed for different columns.}
    Mossel's work was subsequently improved quantitatively by Keller~\cite{Keller}.
    \item \textbf{Intersecting families.} An intersecting family is the same as a polymorphism of the binary predicate $\mathsf{NAND_2}$ on pairs of bits.
    We define approximate polymorphisms of $\mathsf{NAND_2}$ with respect to the unique distribution $\mu_{p,p}$ supported on the support of $\mathsf{NAND}_2$ whose marginals are $\mu_p$. \\ Friedgut and Regev~\cite{FriedgutRegev} proved, in our language, that when $p < 1/2$, any approximate polymorphism of $\mathsf{NAND_2}$ is close to an exact polymorphism of $\mathsf{NAND_2}$.
\end{enumerate}

\subsection{Techniques: \crtCref{thm:main-intro}}

The proof of \Cref{thm:main-intro} is composed of two parts. First, we show that $f$ is close to a junta. Second, we use this to reduce the approximate question to an \emph{exact} question.

We can assume without loss of generality that $g$ depends on all coordinates. If $g$ is an XOR or an NXOR then \Cref{thm:main-intro} reduces to linearity testing (except for the trivial cases $m = 0$ and $m = 1$), and so we can assume that $g$ is not an XOR or an NXOR. This implies that $g$ has an input $\alpha$ with a nonsensitive coordinate $j$, that is, $g(\alpha) = g(\alpha^{\oplus j})$, where $\alpha^{\oplus j}$ results from flipping the $j$th coordinate. To simplify notation, we assume that $j = m$.

\subsubsection{Showing that \texorpdfstring{$f$}{f} is close to a junta}
\paragraph{The basic argument}
In this part, we assume that $(f \circ g^{n})(Z) = (g \circ f^{m})(Z)$ with probability at least $1 - \eta$. Later on, we will choose $\delta$ as a function of both $\eta$ and the size of the junta.

Suppose first that $f$ has no influential variables. In this case, we will show that $f$ is $\epsilon$-close to a constant function. We do this by assuming that $f$ is $\epsilon$-far from constant and reaching a contradiction.

The idea is to construct two correlated inputs $Z,W$, each individually uniformly random, such that
\[
 (f \circ g^{n})(Z) = (f \circ g^{n})(W)
\]
with probability $1$. To sample such $Z,W$, we use the input $\alpha$. Namely, we form $W$ by resampling the $m$th coordinate of each row whose first $m-1$ coordinates agree with $\alpha$. Thus, by the approximate polymorphism condition, it follows 
that $(g\circ f^{m})(Z) = (g\circ f^{m})(W)$ with probability $\geq 1-2\eta$; as we argue next, this last fact will tell us that $f$ must be close to constant.

Since $g$ depends on all coordinates, there is an input $\beta$ such that $g(\beta) \neq g(\beta\oplus e_m)$; suppose without loss of generality that $g(\beta_1,\ldots,\beta_{m-1},x_m) = x_m$. 
By assumption, $f$ is $\epsilon$-far from constant, and so with probability at least $\epsilon^{m-1}$, if we evaluate $f$ on the first $m-1$ columns of $Z$ (which are identical to the corresponding columns of $W$) then we obtain $\beta_1,\ldots,\beta_{m-1}$. When this happens,
\[
 (g \circ f^{m})(Z) = f(\col_m(Z)) \text{ and } 
 (g \circ f^{m})(W) = f(\col_m(W)),
\]
and we get that $f(\col_m(Z)) = f(\col_m(W))$ with probability $\geq 1-2\eta$.
To analyze this event, we consider the following equivalent way of sampling $z = \col_m(Z)$ and $w = \col_m(W)$:
\begin{enumerate}
    \item Sample a subset $R \subseteq [n]$ by including each element with probability $2^{-(m-1)}$; these are the rows whose first $m-1$ columns agree with $\alpha$.
    \item Sample $z_j = w_j$ for each $j \notin R$.
    \item Sample $z_j,w_j$ independently for each $j \in R$.
\end{enumerate}

The first two steps define a random restriction, and to reach a contradiction we would like to argue that this random restriction still has a significant variance with high probability (so that we will in fact have $f(z)\neq f(w)$ with significant probability). Indeed, this is true provided the variance of $f$ is significant and all of the (low-degree) influences of $f$ are small; this is the so-called ``It Ain't Over Till It's Over'' theorem from~\cite{MOO}. A bit more precisely, this result
asserts that provided the influences of $f$ are small, it is extremely likely (the failure probability is smaller than $\epsilon^{m-1}/2$) that $f$ is $\gamma$-far from constant even after the random restriction, where $\gamma(\eps,m)>0$. This gives us
that
\[
 \Pr[(g \circ f^m)(Z) \neq (g \circ f^m)(W)] \geq \left(\epsilon^{m-1} - \frac{\epsilon^{m-1}}{2}\right) \cdot 2\gamma(1-\gamma).
\]
Thus, choosing $\eta$ so that $2\eta$ is smaller than this expression, we reach a contradiction. This contradiction thus implies that if all of the influences of $f$ 
are small, then the only way for $f$ to be an approximate polymorphism of $g$ is that
$f$ is close to a constant.

\paragraph{Lifting the small low-degree influences assumption.} An arbitrary function $f$ could potentially have influential variables. To generalize our argument to this 
case, we make use of a regularity lemma~\cite{Jones} by Jones. This lemma asserts that one may find a small set of variables $T$ such that randomly restricting them in
$f$, one gets a function with no significant low-degree influences with probability 
close to $1$. Thus, we first perform this random restriction, and then use a variant of the above argument to argue that under such restrictions, $f$ must be in fact 
close to a constant. Over all, we obtain that $f$ is close to a junta.

\subsubsection{Deducing \crtCref{thm:main-intro}}
The previous part shows that $f$ is close to a junta $F$, say depending on the first $L$ coordinates. We split the input $Z$ accordingly to two matrices: $Z^{(1)}$ consists of the first $L$ rows, and $Z^{(2)}$ consists of the remaining rows. Thus with probability $1-\delta$,
\begin{multline*}
 f(g(\row_1(Z^{(1)})),\ldots,g(\row_L(Z^{(1)})), g(\row_1(Z^{(2)})),\ldots,g(\row_{n-L}(Z^{(2)}))) = \\ g(f(\col_1(Z^{(1)}),\col_1(Z^{(2)})),\ldots,f(\col_m(Z^{(1)}),\col_m(Z^{(2)}))).
\end{multline*}
If we fix $Z^{(2)}$ then we can find functions $f_0,\ldots,f_m\colon \{0,1\}^L \to \{0,1\}$ such that the left-hand side becomes $f_0 \circ g$, and the right-hand side becomes $g \circ (f_1,\ldots,f_m)$.

For a typical $Z^{(2)}$, the functions $f_1,\ldots,f_m$ are all close to $F$ and so to each other, and furthermore
\[
 \Pr[f_0 \circ g^{n} \neq g \circ (f_1,\ldots,f_m)] \leq \delta.
\]
We choose $\delta < \min(\eta,2^{-mL})$, and so this implies that in fact,
\[
 f_0 \circ g^{n} = g \circ (f_1,\ldots,f_m).
\]

Since the functions $f_1,\ldots,f_m$ are close to each other, the classification of all solutions to the equation $f_0 \circ g = g \circ (f_1,\ldots,f_m)$ implies that $f_1 = \cdots = f_m$ (except for some corner cases), and so $f_0 \circ g = g \circ f_1$. This completes the proof, since both $f_1$ and $f$ are close to $F$.

\subsection{Techniques: \crtCref{thm:list-decoding-intro}}

We illustrate the proof of the theorem in the special case of the AND function. It will be more convenient to switch from $\{0,1\}$ to $\{-1,1\}$, and to consider $g(x_1,x_2) = x_1 \land x_2 = \min(x_1,x_2)$ (which is one possible interpretation of AND).

We prove the theorem in contrapositive: assuming that all low-degree Fourier coefficients are small (which is equivalent to having small correlation with all XORs), we will bound $\Pr[f(x \land y) = f(x) \land f(y)]$, or rather $\E[f(x \land y) (f(x) \land f(y))]$ (this is one reason to switch to $\{-1,1\}$).

The main idea is to apply the invariance principle in order to switch to a problem in Gaussian space. For this, we need $f$ to have small low-degree influences, as well as low degree. We can assume that $f$ has small low-degree influences by appealing to Jones' regularity lemma. In order to reduce the degree of $f$, we apply a small amount of noise. An expansion argument of Mossel~\cite{Mossel2010} shows that the noise doesn't affect the expectation by much, essentially since given $x$ and $x \land y$ there is some uncertainty regarding $y$.

What do we get in Gaussian space? The vector
\[
 \frac{(x \land y) + 1/2}{\sqrt{3/4}}, x, y
\]
has expectation zero and covariance matrix
\[
 \Sigma = \begin{pmatrix}
 1 & \frac{1}{\sqrt{3}} & \frac{1}{\sqrt{3}} \\
 \frac{1}{\sqrt{3}} & 1 & 0 \\
 \frac{1}{\sqrt{3}} & 0 & 1
 \end{pmatrix} .
\]
This shows that for some functions $q,p\colon \mathbb{R}^n \to \mathbb{R}$,
\[
 \E[f(x \land y) (f(x) \land f(y))] \approx \E[q(\cG_0) (p(\cG_1) \land p(\cG_2))],
\]
where $(\cG_0,\cG_1,\cG_2)$ is a multivariate Gaussian with expectation zero and covariance $\Sigma$, and $p(\cG_1) \land p(\cG_2) = (-1 + p(\cG_1) + p(\cG_2) + p(\cG_1) p(\cG_2))/2$ is the multilinear extension of~$\land$. Due to the degree-reducing noise, $p$ depends mostly on the behavior of $f$ around the middle slice, and $q$ depends mostly on its behavior around the quarter slice.
A standard truncation argument lets us assume that $q,p$ attain values in $[-1,1]$, and a further rounding argument lets us assume that they attain values in $\{-1,1\}$.

The assumption that $f$ has no large low-degree Fourier coefficients translates to $\E[p] \approx 0$.\footnote{Due to the application of Jones' regularity lemma, we are actually working with restrictions of $f$ rather than $f$ itself. In order to ensure that these restrictions have expectation close to zero, we need to assume that $f$ has small lower-degree Fourier coefficients.}
We do not have control over $\E[q]$, since it is controlled by the low-degree Fourier coefficients of $f$ with respect to the $\{-1,1\}$-analog of $\mu_{1/4}$. Applying a generalization of Borell's theorem due to Neeman~\cite{Neeman}, this is enough to show that the optimal choice for $p$ is the one-dimensional sign function. By calculating the corresponding optimal choice for $q$, we obtain

% In order to bound $\E[q(\cG_0) p(\cG_1) p(\cG_2)]$, we give another way to sample this distribution:
% \[
%  \cG_0 = \frac{Z_1 + Z_2}{\sqrt{2}}, \cG_1 = \sqrt{\frac23} Z_1 + \sqrt{\frac13} W_1, \cG_2 = \sqrt{\frac23} Z_2 + \sqrt{\frac13} W_2,
% \]
% where $Z_1,Z_2,W_1,W_2$ are independent standard Gaussians.
% By averaging over $W_1,W_2$, using $|q| \leq 1$, and applying Cauchy--Schwarz, we bound our expectation by
% \[
%  \E[|(U_\rho p)(Z_1) \land (U_\rho p)(Z_2)|] \leq \sqrt{\E[((U_\rho p)(Z_1) \land (U_\rho p)(Z_2))^2]}, \quad \rho = \sqrt{2/3}.
% \]
% Expanding the inner expectation using the Fourier expansion of $\land$, we bound the last expression by
% \[
%  \sqrt{\frac{1 + 2\|U_\rho p\|^2 + \|U_\rho p\|^4}{4}} = \frac{1 + \|U_\rho p\|^2}{2}.
% \]
% At this point we invoke Borell's isoperimetric theorem, obtaining a bound of
% \[
%  \frac{1 + \frac{2}{\pi} \arcsin (\rho^2)}{2} \approx 0.732279527198770.
% \]
% This is a bound on the \emph{correlation} of $f(x \land y)$ and $f(x) \land f(y)$. The corresponding bound on the probability that $f(x \land y) = f(x) \land f(y)$ is
% \[
%  \Pr[f(x) \land f(y) = f(x) \land f(y)] \lesssim \frac{1}{2} + \frac{1}{2} \cdot 0.732279527198770 \approx 0.866139763599385.
% \]

% \smallskip

\[
 \Pr[f(x \land y) = f(x) \land f(y)] \lesssim 0.814975356673002.
\]

Here is a matching construction: take $f$ to be the majority function for inputs $x$ such that $\frac{1}{n} \sum_i x_i \approx 0$, and an appropriate threshold function elsewhere. The construction exploits the fact that random points $x,y \in \{-1,1\}^n$ satisfy $\frac{1}{n} \sum_i x_i, \frac{1}{n} \sum_i y_i \approx 0$ while $\frac{1}{n} \sum_i (x_i \land y_i) \approx -\frac{1}{2}$. Choosing the optimal threshold gives a function $f$ with
\[
 \Pr[f(x \land y) = f(x) \land f(y)] \gtrsim 0.814975356673002.
\]

\subsection{Techniques: Classifying exact (multi-)polymorphisms}

Dokow and Holzman~\cite{DH09} classified all exact multi-polymorphisms, that is, all exact solutions to the equation $f_0 \circ g^n = g \circ (f_1,\ldots,f_m)$, using combinatorial arguments. We present an alternative proof using Boolean function analysis in \Cref{sec:exact}. For the sake of the proof, we switch from $\{0,1\}$ to $\{-1,1\}$.

The proof proceeds in two main steps. In the first step, we determine all multilinear polynomials $g,h\colon \{-1,1\}^m \to \mathbb{R}$ and $f_0,\ldots,f_m\colon \{-1,1\}^n \to \mathbb{R}$ which solve the equation
\[
 f_0(g(z_{11},\ldots,z_{1m}),\ldots,g(z_{n1},\ldots,z_{nm})) =
 h(f_1(z_{11},\ldots,z_{n1}),\ldots,f_m(z_{1m},\ldots,z_{nm})),
\]
where the functions $f_0,\ldots,f_m,g,h$ are extended to $\mathbb{R}^n$ or $\mathbb{R}^m$ multilinearly.
Except for some corner cases, these solutions all involve functions of the form
\[
 A \prod_{i \in S} (x_i + \kappa_i) - B.
\]

In the second step, we observe that a function of the form above is Boolean iff it corresponds to either XOR, NXOR, AND, or OR, which completes the classification.

The first step is itself composed of two substeps. In the first substep, we relate the supports of the Fourier expansions of $g,h,f_0,\ldots,f_m$ to that of $f_0 \circ g$ and $g \circ (f_1,\ldots,f_m)$, and conclude that except for some corner cases, and after possibly removing irrelevant coordinates, $\deg g = \deg h = m$ and $\deg f_0 = \deg f_1 = \cdots = \deg f_m = n$. In the second step, we show that up to affine shifts, the only solution to $f_0 \circ g^n = h \circ (f_1,\ldots,f_m)$ is $g(y) = h(y) = \prod_{j=1}^m y_j$ and $f_0(x) = f_1(x) = \cdots = f_m(x) = \prod_{i=1}^m x_i$.

\paragraph{Paper organization} After a few preliminaries in \Cref{sec:prel}, we prove \Cref{thm:main-intro} in \Cref{sec:approximate-polymorphisms}, and its generalization to the setting of several $f$'s in \Cref{sec:approximate-polymorphisms-multi}. We discuss the list decoding regime in \Cref{sec:list-decoding}. \Cref{sec:exact} determines all solutions to the equation $f_0 \circ g^n = h \circ (f_1,\ldots,f_m)$. We close the paper with \Cref{sec:open-questions}, which poses several open questions.

\section{Preliminaries}
\label{sec:prel}

We assume that the reader is familiar with the rudiments of Boolean function analysis, as described in the monograph~\cite{ODonnell}.

We use $[n]$ for $\{1,\ldots,n\}$. If $v$ is a vector, $v|_I$ denotes its restriction to the coordinates in $I$.

The low-degree influence is defined by $\Inf_i^{\leq d}[f] = \Inf_i[f^{\leq d}]$. The distribution $\mu_p$ is the product distribution over $\{0,1\}^n$ in which each coordinate equals~$1$ with probability $p$. We sometimes denote this distribution by $\mu_p(\{0,1\}^n)$. The distribution $\stdnormal$ is a standard normal distribution (zero mean, unit variance).

We use the notation $x^{\oplus n}$ for the $n$-fold XOR of $x$. That is, $x^{\oplus n} = 0$ if $x = 0$ or $n$ is even, and $1^{\oplus n} = 1$ if $n$ is odd. The latter definition makes sense even when $n$ is negative.

A function $f\colon \{0,1\}^n \to \{0,1\}$ is \emph{odd} if $f(\lnot x) = \lnot f(x)$, where $\lnot x = (\lnot x_1,\ldots,\lnot x_n)$, and \emph{even} if $f(\lnot x) = f(x)$. The function $f$ is \emph{balanced} if $\E[f] = 1/2$.

Two Boolean functions $f,g\colon \{0,1\}^n \to \{0,1\}$ are \emph{$\epsilon$-close} if $\Pr[f \neq g] \leq \epsilon$. The functions are \emph{$\epsilon$-close with respect to $\mu_p$} if $\Pr_{\mu_p}[f \neq g] \leq \epsilon$.

A function $f$ is a \emph{$J$-junta} if it depends on at most $J$ coordinates (if $J \in \mathbb{N}$), or if it depends on the coordinates in $J$ (if $J$ is a set).

Unless stated otherwise, random variables ranging over $\{0,1\}^n$ take the uniform distribution $\mu_{1/2}$.

Various results will start by fixing some parameters. Any subsequent big~$O$ bounds will depend on these parameters.

\subsection{Polymorphisms}
\label{sec:prel-polymorphisms}

We start by formally defining polymorphisms. First, we define some required notation.

\begin{definition}[Composition]
Let $f\colon \{0,1\}^n \to \{0,1\}$ and $g\colon \{0,1\}^m \to \{0,1\}$. We define functions $f\circ g^n\colon \{0,1\}^{nm} \to \{0,1\}$ and $g\circ f^m\colon \{0,1\}^{nm} \to \{0,1\}$, whose input is an $n \times m$ matrix $Z$, by
\begin{align*}
(f \circ g^n)(Z) &= f(g(\row_1(Z)),\ldots,g(\row_n(Z))), \\
(g\circ f^m)(Z) &=
g(f(\col_1(Z)),\ldots,f(\col_m(Z))).
\end{align*}

If $f_1,\ldots,f_m\colon \{0,1\}^m \to \{0,1\}$, we also define $g \circ (f_1,\ldots,f_m)\colon \{0,1\}^{nm} \to \{0,1\}$ by
\[
 (g \circ (f_1,\ldots,f_m))(Z) =
 g(f_1(\col_1(Z)), \ldots, f_m(\col_m(Z))).
\]
\end{definition}

We can now define polymorphisms and their generalization, skew-polymorphisms and multi-polymorphisms.

\begin{definition}[Polymorphisms]
A function $f\colon \{0,1\}^n \to \{0,1\}$ is a \emph{polymorphism} of a function $g\colon \{0,1\}^m \to \{0,1\}$ if
\[
 f \circ g^n = g \circ f^m.
\]

The pair $f_0,f_1\colon \{0,1\}^n \to \{0,1\}$ forms a \emph{skew polymorphism} of a function $g\colon \{0,1\}^m \to \{0,1\}$ if
\[
 f_0 \circ g^n = g \circ f_1^m.
\]

A tuple $f_0,\ldots,f_m\colon \{0,1\}^n \to \{0,1\}$ forms a \emph{multi-polymorphism} of $g\colon \{0,1\}^m \to \{0,1\}$ if
\[
 f_0 \circ g^n = g \circ (f_1,\ldots,f_m).
\]

The function $f$ is an \emph{$\epsilon$-approximate} polymorphism of $g$ if
\[
 \Pr[f \circ g^n = g \circ f^m] \geq 1 - \epsilon.
\]
We define approximate versions of the other notions in a similar way.
\end{definition}

Dokow and Holzman~\cite{DH09} essentially characterized all multi-polymorphisms (a result which we reprove in this paper). Let us first state the resulting characterization of polymorphisms.

\begin{theorem} \label{thm:exact-polymorphisms}
If $f\colon \{0,1\}^n \to \{0,1\}$ is a polymorphism of $g\colon \{0,1\}^m \to \{0,1\}$ then one of the following cases holds:
\begin{itemize}
    \item $g = b$ is constant and $f(b,\ldots,b) = b$, or $f = b$ is constant and $g(b,\ldots,b) = b$.
    \item $g = x_i$ and $f$ is arbitrary, or $f = x_i$ and $g$ is arbitrary.
    \item $g = \lnot x_i$ and $f$ is odd, or $f = \lnot x_i$ and $g$ is odd.
    \item There exist $I \subseteq [n]$ and $J \subseteq [m]$ such that $f(x) = \bigoplus_{i \in I} x_i \oplus a$, $g(x) = \prod_{j \in J} x_j \oplus b$, and $a^{\oplus (|J|-1)} = b^{\oplus (|I|-1)}$.
    \item There exist $I \subseteq [n]$ and $J \subseteq [m]$ such that $f(x) = \biglor_{i \in I} x_i$ and $g(x) = \biglor_{j \in J} x_j$.
    \item There exist $I \subseteq [n]$ and $J \subseteq [m]$ such that $f(x) = \bigland_{i \in I} x_i$ and $g(x) = \bigland_{j \in J} x_j$.
\end{itemize}
\end{theorem}

The characterization of skew polymorphisms has a few more cases.

\begin{theorem} \label{thm:exact-skew-polymorphisms}
If $f_0,f_1\colon \{0,1\}^n \to \{0,1\}$ is a skew polymorphism of $g\colon \{0,1\}^m \to \{0,1\}$ then one of the following cases holds:
\begin{itemize}
    \item $g(x) = b$ is constant and $f_0(b,\ldots,b) = b$, or $f_0 = a_0$ and $f_1 = a_1$ are constants and $g(a_1,\ldots,a_1) = a_0$.
    \item $g(x) = x_j$ and $f_0 = f_1$, or $f_0(x) = f_1(x) = x_i$.
    \item $g(x) = \lnot x_j$ and $f_0(\lnot x) = \lnot f_1(x)$.
    \item $f_0(x) = f_1(x) = \lnot x_i$ and $g$ is odd.
    \item $f_0(x) = x_i$, $f_1(x) = \lnot x_i$, and $g$ is even.
    \item There exist $I \subseteq [n]$ and $J \subseteq [m]$ such that $f_0(x) = \bigoplus_{i \in I} x_i \oplus a_0$, $f_1(x) = \bigoplus_{i \in I} x_i \oplus a_1$, $g(x) = \bigoplus_{j \in J} x_j \oplus b$, and $a_0 \oplus b^{\oplus (|I|-1)} = a_1^{\oplus |J|}$.
    \item There exist $I \subseteq [n]$ and $J \subseteq [m]$ such that $f_0(x) = f_1(x) = \biglor_{i \in I} x_i$ and $g(x) = \biglor_{j \in J} x_j$.
    \item There exist $I \subseteq [n]$ and $J \subseteq [m]$ such that $f_0(x) = \biglor_{i \in I} x_i$, $f_1(x) = \bigland_{i \in I} x_i$, and $g(x) = \biglor_{j \in J} \lnot x_j$.
    \item Same as preceding two cases, with $\biglor$ and $\bigland$ switched.
\end{itemize}
\end{theorem}

The characterization of multi-polymorphisms has even more cases.

\begin{theorem}
\label{thm:exact-multi-polymorphisms}
If $f_0,\ldots,f_m\colon \{0,1\}^n \to \{0,1\}$ is a multi-polymorphism of $g\colon \{0,1\}^m \to \{0,1\}$ then one of the following cases holds:
\begin{itemize}
    \item $g = b$ is constant, and $f_0(b,\ldots,b) = b$.
    \item $g = x_i$ and $f_0 = f_i$.
    \item $g = \lnot x_i$ and $f_0(\lnot x) = \lnot f_i(x)$.
    \item $f_0 = b_0$ is constant, and there exists $J \subseteq [m]$ such that $f_j = b_j$ is constant for $j \in J$, and $g(x) = b_0$ whenever $x_j = b_j$ for all $j \in J$.
    \item $f_j = x_i \oplus a_j$ for $j = 0$ and all $j$ that $g$ depends on, where $g(x_1 \oplus a_1,\ldots,x_m \oplus a_m) = g(x_1,\ldots,x_m) \oplus a_0$.
    \item There exist $I \subseteq [n]$ and $J \subseteq [m]$ such that $f_j(x) = \bigoplus_{i \in I} x_i \oplus a_j$ for $j = 0$ and $j \in J$, $g(x) = \bigoplus_{j \in J} x_j \oplus b$, and $a_0 \oplus b^{\oplus (|I|-1)} = \bigoplus_{j \in J} a_j$.
    \item There exist $I \subseteq [n]$ and $J \subseteq [m]$ such that $f_0(x) = \biglor_{i \in I} x_i$, $g(x) = \biglor_{j \in J} (x_j \oplus a_j)$, and for every $j \in J$, if $a_j = 0$ then $f_j(x) = \biglor_{i \in I} x_i$, and if $a_j = 1$ then $f_j(x) = \bigland_{i \in I} x_i$.
    \item Same as preceding case, with $\biglor$ and $\bigland$ switched.
\end{itemize}
\end{theorem}

\subsection{Results from Boolean function analysis}
\label{sec:prel-bfa}

We will need several results concerning random restrictions and regular functions, which we first define.

\begin{definition}[Restriction]
Let $f\colon \{0,1\}^n \to \{0,1\}$, let $J \subseteq [n]$, and let $z \in \{0,1\}^{\overline{J}}$. We define the \emph{restriction} $f_{\overline{J}\to z}\colon \{0,1\}^J \to \{0,1\}$ by
\[
 f_{\overline{J} \to z}(y) = f(x),
 \text{ where }
 x_j = \begin{cases}
 z_j & \text{if } j \notin J, \\
 y_j & \text{if } j \in J.
 \end{cases}
\]
\end{definition}

\begin{definition}[Random restriction]
A \emph{$p$-random restriction} is a pair $(J,z)$, where $J \subseteq [n]$ (the parameter $n$ will be clear from context) and $z \in \{0,1\}^{\overline{J}}$ are sampled as follows: every element in $[n]$ is included in $J$ with probability $p$, and $z$ is chosen uniformly at random.
We denote this distribution by $\rr{p}$.
\end{definition}

\begin{definition}[Regularity]
A Boolean function $f\colon \{0,1\}^n \to \{0,1\}$ is \emph{($d,\tau)$-regular} if $\Inf_i^{\leq d}[f] \leq \tau$ for all $i \in [n]$.
\end{definition}

The first result is a regularity lemma due to Jones~\cite{Jones}; see also~\cite[Lemma 2.7]{Mossel2020}.

\begin{theorem}[\cite{Jones}]
\label{thm:jones}

For all $d \in \mathbb{N}$ and $\tau,\epsilon > 0$ there exists $L \in \mathbb{N}$ such that the following holds.

For every $f\colon \{0,1\}^n \to \{0,1\}$ we can find a set $T \subseteq [n]$ of at most $L$ coordinates such that
\[
 \Pr_{z \in \{0,1\}^T}[f_{T \to z} \text{ is $(d,\tau)$-regular}] \geq 1-\epsilon.
\]

Furthermore, the result holds with respect to any finite number of product measures $\mu_{p_1},\ldots,\mu_{p_\ell}$. That is, we can find a set $T$ such that for each $i \in [\ell]$, the statement above holds when $z$ is sampled according to $\mu_{p_i}$ and when regularity is defined with respect to $\mu_{p_i}$.
\end{theorem}

We explain how to deduce this result from Jones' argument in \Cref{apx:jones}.

The second result is the famous ``It Ain't Over Till It's Over'' theorem.

\begin{theorem}[\cite{MOO}] \label{thm:it-aint-over}

For all $\epsilon_1,\epsilon_2,p > 0$ there exist $d \in \mathbb{N}$ and $\tau,\delta > 0$ such that if $f\colon \{0,1\}^n \to \{0,1\}$ is $(d,\tau)$-regular and $\V{f} \geq \epsilon_1$ then
\[
 \Pr_{(J,z) \sim \rr{p}}[\V{f_{\overline{J}\to z}} \geq \delta] \geq 1-\epsilon_2.
\]
\end{theorem}

The third result states that noise-stable functions are close to constant.

\begin{definition}[Noise sensitivity] \label{def:noise-sensitivity}
Let $p \in (0,1)$ and $\rho \in [0,1]$, and let $f\colon \{0,1\}^n \to \{0,1\}$. The \emph{noise sensitivity} of $f$ with respect to $\mu_p$ is
\[
 \NS_\rho^{\mu_p}(f) = \Pr_{x,y \sim N_\rho^{\mu_p}}[f(x) \neq f(y)],
\]
where $N_\rho^{\mu_p}$ is the following distribution: we sample $x \sim \mu_p$, and for each $i \in [n]$ independently, with probability $\rho$ we let $y_i = x_i$, and otherwise we sample $y_i$ according to $\mu_p$.
\end{definition}

\begin{lemma} \label{lem:noise-stable-constant}
Fix $p,\rho \in (0,1)$. If $f\colon \{0,1\}^n \to \{0,1\}$ satisfies $\NS_\rho^{\mu_p}(f) \leq \epsilon$ then $f$ is $O(\epsilon)$-close to a constant function.
\end{lemma}
\begin{proof}
Define $F\colon \{0,1\}^n \to \{-1,1\}$ by $F(x) = (-1)^{f(x)}$.
On the one hand,
\[
 \E_{x,y \sim N_\rho^{\mu_p}}[F(x) F(y)] = 1 - 2\NS_\rho^{\mu_p}(f) \geq 1 - 2\epsilon,
\]
On the other hand, it is well-known that
\[
 \E_{x,y \sim N_\rho^{\mu_p}}[F(x) F(y)] = \sum_{d=0}^n \rho^d \|F^{=d}\|^2 \leq (1 - \V{F}) + \rho \V{F} = 1 - (1-\rho) \V{F}.
\]
It follows that $\V{F} = O(\epsilon)$, and so $f$ is $O(\epsilon)$-close to a constant function.
\end{proof}

\section{Approximate polymorphisms}
\label{sec:approximate-polymorphisms}

In this section we prove the following result, which implies \Cref{thm:main-intro} by simple case analysis.

\begin{theorem} \label{thm:main}
Fix $g\colon \{0,1\}^m \to \{0,1\}$, and let $p = \Pr[g = 1]$. For every $\epsilon > 0$ there exists $\delta > 0$ (depending on both $g$ and $\epsilon$) such that if $f\colon \{0,1\}^n \to \{0,1\}$ is a $\delta$-approximate polymorphism of $g$ then there exists a skew polymorphism $f_0,f_1$ of $g$ such that $f$ is $\epsilon$-close to $f_1$ with respect to $\mu_{1/2}$ and to $f_0$ with respect to $\mu_p$.

If $g$ is balanced then $f$ is moreover $\epsilon$-close to a polymorphism of $g$.
\end{theorem}

We can assume without loss of generality that $g$ depends on all coordinates (we show this formally below). If $g$ is an XOR or NXOR, then the result is either trivial (if $g$ depends on at most one coordinate), or follows from an analysis in the style of linearity testing (if $g$ depends on at least two coordinates). Therefore we concentrate on the case in which $g$ depends on all coordinates and is neither an XOR nor an NXOR.

When $g$ is neither an XOR nor an NXOR, it must have a non-trivial certificate $\alpha$.

\begin{lemma} \label{lem:g-alpha-beta}
Suppose that $g\colon \{0,1\}^m \to \{0,1\}$ depends on all coordinates, and is not of the form $\bigoplus_{i=1}^m x_i \oplus b$.

After possibly reordering the coordinates of $g$, we can find $\alpha$ such that
\[
 g(\alpha_1,\ldots,\alpha_{m-1},0) = g(\alpha_1,\ldots,\alpha_{m-1},1).
\]
\end{lemma}
\begin{proof}
If it is impossible to find such $\alpha$ even after reordering the coordinates, then for every $x \in \{0,1\}^m$ and any $j \in [m]$ we have $g(x) \neq g(x^{\oplus j})$, where $x^{\oplus j}$ is obtained from $x$ by flipping the $j$th coordinate. Thus
\[
 g(x_1,\ldots,x_m) = g(0,\ldots,0) \oplus \bigoplus_{j=1}^m x_j,
\]
in contrast to the assumption. Therefore we can find the desired $\alpha$ after possibly reordering the coordinates.
\end{proof}

We can use the certificate $\alpha$ to show that if $f$ is an approximate polymorphism of $g$, then $f$ is close to a junta.

\begin{lemma} \label{lem:f-junta}
Fix a function $g\colon \{0,1\}^m \to \{0,1\}$ which depends on all coordinates and is not of the form $\bigoplus_{i=1}^m x_i \oplus b$.

For every $\epsilon > 0$ there exist $L \in \mathbb{N}$ and $\eta > 0$ such that any $\eta$-approximate polymorphism of $g$ is $\epsilon$-close to an $L$-junta.
\end{lemma}

We prove \Cref{lem:f-junta} in \Cref{sec:approximate-polymorphisms-junta}. Let us see how it implies \Cref{thm:main}.

\Cref{lem:f-junta} requires $g$ to depend on all coordinates. Let us first dispense with this requirement.

\begin{lemma} \label{lem:approximate-polymogrphism-all-coords}
Suppose that \Cref{thm:main} holds when $g$ depends on all coordinates. Then it holds for all $g$.
\end{lemma}
\begin{proof}
If $g$ depends only on the coordinates in some subset $J \subseteq [m]$, we can find a function $G\colon \{0,1\}^J \to \{0,1\}$ such that $g(x) = G(x|_J)$. By assumption, \Cref{thm:main} holds for $G$, say with some parameter $\delta$. We will show that \Cref{thm:main} holds for $g$ as well, with the same parameter $\delta$.

If $f$ is a $\delta$-approximate polymorphism of $g$, then it is also a $\delta$-approximate polymorphism of $G$. \Cref{thm:main} shows that $f$ is $\epsilon$-close (in an appropriate sense) to a skew polymorphism of $G$, which is also a skew polymorphism of $g$.
\end{proof}

\Cref{lem:f-junta} also requires $g$ not to be XOR or its negation. We first show how to prove \Cref{thm:main} under this restriction, and then handle all other cases. 
\begin{lemma} \label{lem:approximate-polymorphisms-not-xor}
\Cref{thm:main} holds when $g$ depends on all coordinates and is not of the form $\bigoplus_{i=1}^m x_i \oplus b$.
\end{lemma}
\begin{proof}
Given $\epsilon > 0$ (which we assume is small enough), let $L \in \mathbb{N}$ and $\eta > 0$ be the constants promised by \Cref{lem:f-junta}. We choose $\delta = \min(\epsilon,\eta,2^{-mL}/3)$.

Let $f$ be a $\delta$-approximate polymorphism of $g$. According to \Cref{lem:f-junta}, $f$ is $\epsilon$-close to an $L$-junta $F$, say depending on the coordinates $I \subseteq [n]$. We also think of $F$ as a function $F_I\colon \{0,1\}^I \to \{0,1\}$.

Let $W$ be a random $\overline{I} \times m$ zero-one matrix. With probability at least $2/3$ over the choice of $W$,
\[
 \Pr_Z[(f \circ g^n)(Z) \neq (g \circ f^m)(Z) \mid \row_i(Z) = \row_i(W) \text{ for all } i \in \overline{I}] < 3\delta.
\]
For each $j \in [m]$, with probability at least $1-1/(2m)$ over the choice of $W$,
\[
 \Pr_{x \in \{0,1\}^n}[f(x) \neq F(x) \mid x|_{\overline{I}} = \col_j(W)] < 2m\epsilon.
\]
Since $1/3 + m \cdot 1/(2m) < 1$, we can find $W$ so that all of these $m+1$ events happen.

Using $W$, we define $m+1$ functions $f_0,\ldots,f_m\colon \{0,1\}^I \to \{0,1\}$ which will turn out to be a multi-polymorphism of $g$:
\begin{itemize}
    \item For $j \in [m]$, we let $f_j$ be the restriction of $f$ obtained by fixing the coordinates in $\overline{I}$ to the values $\col_j(W)$.
    \item Similarly, $f_0$ is the restriction of $f$ obtained by fixing coordinate $i \notin I$ to $g(\row_i(W))$.
\end{itemize}
By construction,
\[
 \Pr_{Z \in \{0,1\}^{I \times [m]}}[f_0 \circ g^n \neq g \circ (f_1,\ldots,f_m)] < 3\delta \leq 2^{-mL},
\]
and so $f_0,\ldots,f_m$ is a multi-polymorphism of~$g$.

Since all of $f_1,\ldots,f_m$ are $O(\epsilon)$-close to $F_I$, all of them are $O(\epsilon)$-close to each other, and moreover $f$ is $O(\epsilon)$-close to all of them. 

According to \Cref{thm:exact-multi-polymorphisms}, one of the following cases holds:
\begin{enumerate}
    \item $f_0 = b_0$ is constant, and there exists $J \subseteq [m]$ such that $f_j = b_j$ is constant for $j \in J$, and $g(x) = b_0$ whenever $x_j = b_j$ for all $j \in J$.
    \item $f_j = x_i \oplus a_j$ for some $i \in I$ and $j \in \{0,\ldots,m\}$.
    \item $g(x) = \biglor_{j=1}^m (x_j \oplus a_j)$, and there exists $K \subseteq I$ such that $f_0(x) = \biglor_{i \in K} x_i$ and for $j \in [m]$, if $a_j = 0$ then $f_j(x) = \biglor_{i \in K} x_i$, and if $a_j = 1$ then $f_j(x) = \bigland_{i \in K} x_i$.
    \item Same as preceding case, with $\biglor$ and $\bigland$ switched.
\end{enumerate}

If the first case holds then $J \neq \emptyset$ since $g$ is non-constant, and so $f$ is $O(\epsilon)$-close to some constant $a_1$. Let $a_0 = g(a_1,\ldots,a_1)$. With probability $1 - O(\epsilon)$, $(f \circ g^n)(Z) = (g \circ f^m)(Z) = g(a_1,\ldots,a_1) = a_0$, and so $f$ is $O(\epsilon)$-close to $a_0$ with respect to $\mu_p$.

In all other cases, we observe that since $f_1,\ldots,f_m$ are $O(\epsilon)$-close to each other, necessarily $f_1 = \cdots = f_m$, and so $f_0,f_1$ is a skew polymorphism of $g$. This is clear in the second case. In the third and fourth cases, we can assume that $|K| \geq 2$, since otherwise we are in one of the first two cases. Then $f_1 = \cdots = f_m$ follows from
\[
 \Pr\left[ \biglor_{i \in K} x_i = \bigland_{i \in K} x_i \right] = \frac{2}{2^{|K|}} \leq \frac{1}{2}.
\]
Clearly $f$ is $O(\epsilon)$-close to $f_1$. It remains to show that $f$ is $O(\epsilon)$-close to $f_0$ with respect to $\mu_p$. To see this, note that with probability $1 - O(\epsilon)$,
\[
 (f \circ g^n)(Z) = (g \circ f^m)(Z) = (g \circ f_1^m)(Z) = (f_0 \circ g^n)(Z).
\]

Finally, when $g$ is balanced, $f_0$ and $f_1$ are $O(\epsilon)$-close, and so a simple case analysis shows that $f_0 = f_1$. Therefore $f_1$ is a polymorphism of $g$.
\end{proof}

To complete the proof of \Cref{thm:main}, we show that it holds when $g$ is an XOR or its negation. This includes the cases that $g$ is constant ($m = 0$) or an (anti-)dictator ($m = 1$). When $m \neq 0$, the function $g$ is balanced, and so it suffices to show that any approximate polymorphism of $g$ is close to an exact polymorphism of $g$.

\begin{lemma} \label{lem:approximate-polymorphism-xor}
\Cref{thm:main} holds when $g(x) = \bigoplus_{i=1}^m x_i \oplus b$.
\end{lemma}
\begin{proof}
If $m = 0$ then $g = b$, and so a $\delta$-approximate polymorphism $f$ of $g$ satisfies
\[
 \Pr[f(b,\ldots,b) = b] \geq 1-\delta.
\]
For any $\delta < 1$ this implies that $f(b,\ldots,b) = b$, and so $f$ is a polymorphism of $g$.

If $m = 1$ then $g = x_j$ or $g = \lnot x_j$. If $g = x_j$ then every function is a polymorphism of $g$. If $g = \lnot x_j$ then a $\delta$-approximate polymorphism $f$ of $g$ satisfies
\[
 \Pr[f(\lnot x_1,\ldots,\lnot x_n) = \lnot f(x_1,\ldots,x_n)] \geq 1-\delta.
\]
In particular, if we define $F\colon \{0,1\}^n \to \{0,1\}$ by
\[
 F(x) =
 \begin{cases}
 f(x) & \text{if } x_j = 0, \\
 \lnot f(\lnot x) & \text{if } x_j = 1,
 \end{cases}
\]
then $F$ is odd and
\[
 \Pr[f \neq F] =
 \frac{1}{2} \Pr[f(x) \neq F(x) \mid x_j = 1] =
 \frac{1}{2} \Pr[f(x) \neq \lnot f(\lnot x) \mid x_j = 1] \leq \delta.
\]
Therefore, choosing $\delta = \epsilon$, we get that $f$ is $\epsilon$-close to a polymorphism of $g$.

Finally, suppose that $m \geq 2$. If $f$ is a $\delta$-approximate polymorphism of $g$ then the function $F\colon \{-1,1\}^n \to \{-1,1\}$, defined by $F((-1)^{x_1},\ldots,(-1)^{x_n}) = (-1)^{f(x_1,\ldots,x_n)}$, satisfies
\[
 \E_{z \in \{-1,1\}^{nm}}\left[
 F\left((-1)^b \prod_{j=1}^m z_{1j}, \ldots, (-1)^b \prod_{j=1}^m z_{nj}\right) \cdot
 (-1)^b \prod_{j=1}^m F(z_{1j},\ldots,z_{nj})
 \right] \geq 1 - 2\delta.
\]
Expanding the left-hand side according to the Fourier expansion of $F$, we obtain
\[
 \sum_{S_0,\ldots,S_m \subseteq [n]} \hat{F}(S_0) \cdots \hat{F}(S_m) (-1)^{b(|S_0|+1)}
 \prod_{j=1}^m 
 \E\left[\prod_{i \in S_0 \Delta S_j} z_{ij} \right] =
 \sum_{S \subseteq [n]} (-1)^{b(|S|+1)} \hat{F}(S)^{m+1}.
\]
Since $\sum_S \hat{F}(S)^2 = 1$, we conclude that there must exist $S \subseteq [n]$ such that $|\hat{F}(S)|^{m-1} \geq 1 - 2\delta$, which implies that $|\hat{F}(S)| \geq 1 - O(\delta)$, and so $f$ is $O(\delta)$-close to $\bigoplus_{i \in S} x_i \oplus a$, for some $a \in \{0,1\}$. We choose $\delta$ so the closeness guarantee becomes~$\epsilon$.

To complete the proof, we need to show that $b^{\oplus (|S|-1)} = a^{\oplus (m-1)}$. To this end, note that with probability $1-O(\epsilon)$ over $Z$,
\[
 a \oplus \bigoplus_{i \in S} \left( \bigoplus_{j=1}^m Z_{ij} \oplus b\right) =
 b \oplus \bigoplus_{j=1}^m \left( \bigoplus_{i \in S} Z_{ij} \oplus a \right).
\]
Considering any $Z$ for which this holds, we immediately obtain $a \oplus b^{\oplus |S|} = b \oplus a^{\oplus m}$, and so $b^{\oplus (|S|-1)} = a^{\oplus (m-1)}$.
\end{proof}

\subsection{Closeness to junta}
\label{sec:approximate-polymorphisms-junta}

In this section we complete the proof of \Cref{thm:main} by proving \Cref{lem:f-junta}. For future use, we prove a more general version; the reader can mentally take $f_1 = \cdots = f_m = f$.

\begin{lemma} \label{lem:f-junta-general}
Fix a function $g\colon \{0,1\}^m \to \{0,1\}$ which depends on all coordinates, and assume moreover that $g$ is not of the form $\bigoplus_{i=1}^m x_i \oplus b$.

For every $\epsilon > 0$ there exist $L \in \mathbb{N}$ and $\eta > 0$ such that if $f_0,\ldots,f_m$ is an $\eta$-approximate multi-polymorphism of $g$ then $f_j$ is $\epsilon$-close to an $L$-junta for some $j \in [m]$.
\end{lemma}

According to \Cref{lem:g-alpha-beta}, we can reorder the coordinates of $g$ so that some assignment $\alpha \in \{0,1\}^{m-1}$ satisfies
\[
 g(\alpha_1,\ldots,\alpha_{m-1},0) = g(\alpha_1,\ldots,\alpha_{m-1},1).
\]
Since $\h$ depends on all coordinates, we can find an assignment $\beta \in \{0,1\}^{m-1}$ and $b \in \{0,1\}$ so that
\[
 g(\beta_1,\ldots,\beta_{m-1},x_m) = x_m \oplus b.
\]

We will construct the junta by appealing to Jones' regularity lemma (\Cref{thm:jones}). First, we need to determine appropriate parameters:
\begin{itemize}
    \item Let $d,\tau,\delta$ be the parameters promised by \Cref{thm:it-aint-over} for $p = 2^{-(m-1)}$, $\epsilon_1 = \epsilon/4$, and $\epsilon_2 = \epsilon^{m-1}/2$.
    \item Let $L$ be the parameter promised by \Cref{thm:jones} for $d,\tau$ and $\epsilon/5$.
\end{itemize}
We will prove the lemma for $\eta = (\epsilon/20)^m\delta$.

According to \Cref{thm:jones}, we can find a set $T \subseteq [n]$ of at most $L$ coordinates such that
\begin{equation} \label{eq:jones-promise-ap}
 \Pr_{z \in \{0,1\}^T}[\text{$(f_m)_{T \to z}$ is  $(d,\tau)$-regular}] \geq 1-\epsilon/5.
\end{equation}
We will show that this implies that either $f_j$ is $\epsilon$-close to constant for some $j \in [m-1]$, or
\begin{equation} \label{eq:jones-conclusion-ap}
\Pr_{z \in \{0,1\}^T} [\V{(f_m)_{T \to z}} \geq \epsilon/4] \leq \epsilon/4.
\end{equation}
This will complete the proof of \Cref{lem:f-junta-general}.
Indeed, if $f_j$ is $\epsilon$-close to constant for some $j \in [m-1]$ then the lemma clearly hold. If \eqref{eq:jones-conclusion-ap} holds then define a $T$-junta $F$ by letting $F(x)$ be the majority value of $(f_m)_{T \to x|_T}$. We also think of $F$ as a function $F_T\colon \{0,1\}^T \to \{0,1\}$. Since
\[
 \V{(f_m)_{T \to z}} = \Pr[(f_m)_{T \to z} = F_T(z)] \Pr[(f_m)_{T \to z} \neq F_T(z)] \geq \tfrac{1}{2} \Pr[(f_m)_{T \to z} \neq F_T(z)],
\]
we see that
\[
 \Pr[f_m \neq F] = \E_{z \in \{0,1\}^T} \bigl[ \Pr[(f_m)_{T \to z} \neq F_T(z)] \bigr] \leq
 2\E_{z \in \{0,1\}^T}[\V{(f_m)_{T \to z}}] \stackrel{\eqref{eq:jones-conclusion-ap}}\leq
 2 \left(\frac{\epsilon}{4} \cdot 1 + \left(1 - \frac{\epsilon}{4}\right) \cdot \frac{\epsilon}{4}\right) \leq \epsilon.
\]

In the rest of this section, we show that \eqref{eq:jones-promise-ap} implies that either $f_j$ is $\epsilon$-close to constant for some $j \in [m-1]$, or \eqref{eq:jones-conclusion-ap} holds. The proof is by contradiction: we assume that \eqref{eq:jones-promise-ap} holds and $f_j$ is $\epsilon$-far from constant for all $j \in [m-1]$, and \eqref{eq:jones-conclusion-ap} fails, and reach a contradiction.

\medskip

We construct a pair $(Z,W)$ of coupled $n \times m$ zero-one matrices in the following way:
\begin{itemize}
    \item Choose $Z$ uniformly at random.
    \item Let $R \subseteq \overline{T}$ consist of those rows $i \notin T$ such that $\row_i(Z)|_{[m-1]} = \alpha$.
    \item Define $W$ to be the matrix obtained from $Z$ by resampling $W_{im}$ for $i \in R$.
\end{itemize}

By construction,
\[
 (f_0 \circ g^n)(Z) = (f_0 \circ g^n)(W),
\]
and so
\[
 \Pr_{Z,W}[(\h \circ (f_1,\ldots,f_m))(Z) \neq (\h \circ (f_1,\ldots,f_m))(W)] \leq 2\eta.
\]
In the rest of this section, we give a lower bound on the probability that $(\h \circ (f_1,\ldots,f_m))(Z) \neq (\h \circ (f_1,\ldots,f_m))(W)$ which is larger than $2\eta$, thus reaching a contradiction.

We will sample $Z,W$ in three stages:

\begin{enumerate}
    \item First, we sample the $m$ column of the rows in $T$.
    \item Second, we sample columns $1,\ldots,m-1$, which defines $R$, and column $m$ of the rows in $\overline{T} \setminus R$.
    \item Third, we sample column $m$ of the rows in $R$. This is the only part in which $Z,W$ differ.
\end{enumerate}

We say that the first stage is successful if $(f_m)_{T \to \col_m(Z)|_T}$ has variance at least $\epsilon/4$ and is $(d,\tau)$-regular. Our assumptions implies that this holds with probability at least
\[
 \epsilon/4 - \epsilon/5 = \epsilon/20.
\]

Suppose now that the first stage is successful.
We say that the second stage is successful if (i) $f_j(\col_j(Z)) = \beta_j$ for all $j \in [m-1]$, and (ii) $(f_m)_{\overline{R} \to \col_m(Z)|_{\overline{R}}}$ has variance at least $\delta$. Property~(i) holds with probability at least $\epsilon^{m-1}$, since by assumption, $f_j$ is $\epsilon$-far from constant for all $j \in [m-1]$. Since each row in $\overline{T}$ belongs to $R$ with probability $2^{-(m-1)}$, by \Cref{thm:it-aint-over}, property~(ii) fails with probability at most $\epsilon^{m-1}/2$. Hence the first two stages are successful with probability at least
\[
 (\epsilon/20) \cdot \epsilon^{m-1}/2 > (\epsilon/20)^m.
\]

If the first two stages are successful then
\begin{align*}
 (\h \circ (f_1,\ldots,f_m))(Z) &= \h(\beta_1,\ldots,\beta_{m-1},f_m(\col_m(Z))) = f_m(\col_m(Z)) \oplus b, \\
 (\h \circ (f_1,\ldots,f_m))(W) &= \h(\beta_1,\ldots,\beta_{m-1},f_m(\col_m(W))) = f_m(\col_m(W)) \oplus b,
\end{align*}
and $(f_m)_{\overline{R} \to \col_m(Z)|_{\overline{R}}}$ has variance at least $\delta$. Hence $(\h \circ (f_1,\ldots,f_m))(Z) \neq (\h \circ (f_1,\ldots,f_m))(W)$ with probability at least $2\delta$. Consequently,
\[
 \Pr[(\h \circ (f_1,\ldots,f_m))(Z) \neq (\h \circ (f_1,\ldots,f_m))(W)] > 2(\epsilon/20)^m\delta \geq 2\eta,
\]
and we reach a contradiction.

\section{Approximate multi-polymorphisms}
\label{sec:approximate-polymorphisms-multi}

In this section we prove the following generalization of \Cref{thm:main}.

\begin{theorem} \label{thm:main-multi}
Fix $g\colon \{0,1\}^m \to \{0,1\}$, and let $p = \Pr[g=1]$. For every $\epsilon > 0$ there exists $\delta > 0$ (depending on both $g$ and $\epsilon$) such that if $f_0,\ldots,f_m\colon \{0,1\}^n \to \{0,1\}$ is a $\delta$-approximate mult-polymorphism of $g$ then there exists an exact multi-polymorphism $F_0,\ldots,F_m$ of $g$ such that $f_0$ is $\epsilon$-close to $F_0$ with respect to $\mu_p$, and $f_1,\ldots,f_m$ are $\epsilon$-close to $F_1,\ldots,F_m$ with respect to $\mu_{1/2}$.
\end{theorem}

We will require an analog of \Cref{lem:f-junta}, which we prove in \Cref{sec:approximate-multi-polymorphisms-junta}.

\begin{lemma} \label{lem:f-junta-multi}
Fix a function $g\colon \{0,1\}^m \to \{0,1\}$ which depends on all coordinates, and assume moreover that $g$ is not of the form $\bigoplus_{i=1}^m x_i \oplus b$. Let $p = \E[g]$.

For every $\epsilon > 0$ there exist $L \in \mathbb{N}$ and $\eta > 0$ such that if $f_0,\ldots,f_m$ is an $\eta$-approximate multi-polymorphism of $g$ then one of the following cases holds:
\begin{enumerate}[(i)]
\item There exists a set $T \subseteq [n]$ of at most $L$ coordinates such that $f_0$ is $\epsilon$-close to a $T$-junta with respect to $\mu_p$, and $f_1,\ldots,f_m$ are $\epsilon$-close to $T$-juntas with respect to $\mu_{1/2}$.
\item There exist $J \subseteq [m]$, $a \in \{0,1\}^J$ and $a_0 \in \{0,1\}$ such that $f_0$ is $\epsilon$-close to $a_0$ with respect to $\mu_p$, $f_j$ is $\epsilon$-close to $a_j$ for all $j \in J$, and $g(x) = a_0$ whenever $x_j = a_j$ for all $j \in J$.
\end{enumerate}
\end{lemma}

An argument analogous to \Cref{lem:approximate-polymogrphism-all-coords} shows that it suffices to prove \Cref{thm:main-multi} when $g$ depends on all coordinates. The proof of \Cref{lem:approximate-polymorphisms-not-xor} carries through, as we briefly spell out.

\begin{lemma} \label{lem:approximate-multi-polymorphisms-not-xor}
\Cref{thm:main-multi} holds when $g$ depends on all coordinates and is not of the form $\bigoplus_{i=1}^m x_i \oplus b$.
\end{lemma}
\begin{proof}
Given $\epsilon > 0$, let $L \in \mathbb{N}$ and $\eta > 0$ be the constants promised by \Cref{lem:f-junta-multi}. We choose $\delta = \min(\epsilon, \eta, 2^{-mL}/3)$.

Apply \Cref{lem:f-junta-multi} to the $\delta$-approximate multi-polymorphism $f_0,\ldots,f_m$ of $g$. In case~(ii), there is nothing more to prove, so assume that case~(i) holds.

As in the proof of \Cref{lem:approximate-polymorphisms-not-xor}, we find restrictions $F_0,\ldots,F_m\colon \{0,1\}^T \to \{0,1\}$ which form a multi-polymorphism of $g$, and such that $F_j$ is $\epsilon$-close to $f_j$. This completes the proof in this case.
\end{proof}

It remains to settle the case where $g$ is an XOR or an NXOR, in which case we need to generalize slightly the argument of \Cref{lem:approximate-polymorphism-xor}.

\begin{lemma} \label{lem:approximate-multi-polymorphisms-xor}
\Cref{thm:main-multi} holds when $g(x) = \bigoplus_{i = 1}^m x_i \oplus b$.
\end{lemma}
\begin{proof}
The proofs of the cases $m = 0$ and $m = 1$ are very similar to the corresponding proofs in \Cref{lem:approximate-polymorphism-xor}, so assume that $m \geq 2$. We will prove the lemma for $\delta = c\epsilon^2$, where $c>0$ is some constant.

Define functions $F_0,\ldots,F_m\colon \{-1,1\}^n \to \{-1,1\}$ by $F_j((-1)^{x_1},\ldots,(-1)^{x_n}) = (-1)^{f_j(x_1,\ldots,x_n)}$. A calculation along the lines of \Cref{lem:approximate-polymorphism-xor} shows that
\[
 \sum_{S \subseteq [n]} (-1)^{b(|S|+1)} \prod_{j=0}^m \hat{F}_j(S) \geq 1 - 2\delta.
\]
We can bound the left-hand side in absolute value by
\[
 \max_{S \subseteq [n]} |\hat{F}_0(S)| \sum_{S \subseteq [n]} \prod_{j=1}^m |\hat{F}_j(S)|.
\]
The generalized H\"older inequality shows that
\[
 \sum_{S \subseteq [n]} \prod_{j=1}^m |\hat{F}_j(S)| \leq
 \sqrt[m]{\prod_{j=1}^m \sum_{S \subseteq [n]} |\hat{F}_j(S)|^m} \leq
 \sqrt[m]{\prod_{j=1}^m \sum_{S \subseteq [n]} |\hat{F}_j(S)|^2} \leq 1,
\]
and so $|\hat{F}_0(S)| \geq 1-2\delta$ for some $S \subseteq [n]$. Parseval's identity shows that $|\hat{F}_0(T)|^2 = O(\delta)$ for all $T \neq S$, and so
\[
 \sum_{\substack{T \subseteq [n] \\ T \neq S}} \prod_{j=0}^m |\hat{F}_j(T)| \leq O(\sqrt{\delta}) \cdot \sum_{T \subseteq [n]} \prod_{j=1}^m |\hat{F}_j(T)| = O(\sqrt{\delta}),
\]
implying that
\[
 (-1)^{b(|S|+1)} \prod_{j=0}^m \hat{F}_j(S) \geq 1 - O(\sqrt{\delta}).
\]
In particular, $|\hat{F}_j(S)| \geq 1 - O(\sqrt{\delta})$ for all $j \in \{0,\ldots,m\}$, and so each $f_j$ is $O(\sqrt{\delta})$-close to $\bigoplus_{i \in S} x_i \oplus a_j$. We complete the proof by showing that $a_0 \oplus b^{\oplus (|S|-1)} = \bigoplus_{j=1}^m a_j$, using the same argument that concludes \Cref{lem:approximate-polymorphism-xor}.
\end{proof}

\subsection{Closeness to junta}
\label{sec:approximate-multi-polymorphisms-junta}

In this section we complete the proof of \Cref{thm:main-multi} by proving \Cref{lem:f-junta-multi}.

We break up the proof into four parts:
\begin{enumerate}
    \item \Cref{lem:f-junta-general} shows that \emph{one} of $f_1,\ldots,f_m$ is close to a junta.
    \item This implies that $f_0$ is close to a junta (with respect to $\mu_p$).
    \item We conclude that \emph{all} of $f_1,\ldots,f_m$ are close to juntas, unless some of $f_1,\ldots,f_m$ are close to constants.
    \item In the latter case, we show that $f_0$ is close to a constant $a_0$, and those of $f_1,\ldots,f_m$ which are close to constants constitute an $a_0$-certificate for $g$.
\end{enumerate}

We start by showing that $f_0$ is close to a junta.

\begin{lemma} \label{lem:f-junta-multi-f0}
Fix a function $g\colon \{0,1\}^m \to \{0,1\}$ which depends on all coordinates, and assume moreover that $g$ is not of the form $\bigoplus_{i=1}^m x_i \oplus b$. Let $p = \E[g]$.

For every $\epsilon > 0$ there exist $L \in \mathbb{N}$ and $\eta > 0$ such that if $f_0,\ldots,f_m$ is an $\eta$-approximate multi-polymorphism of $g$ then $f_0$ is $\epsilon$-close to an $L$-junta with respect to $\mu_p$.
\end{lemma}
\begin{proof}
Given $\epsilon > 0$, let $L,\eta_0$ be the parameters promised by \Cref{lem:f-junta-general}, and let $\eta = \min(\eta_0,\epsilon)$. According to the lemma, there exists $j \in [m]$ such that $f_j$ is $\epsilon$-close to an $L$-junta. Without loss of generality, suppose that $f_m$ is close to a $T$-junta, where $T \subseteq [n]$ contains at most $L$ coordinates.

We construct a pair $(Z,W)$ of coupled $n \times m$ zero-one matrices as follows: choose $Z$ uniformly at random, and let $W$ be the matrix obtained from $Z$ by resampling $W_{im}$ for $i \notin T$.

With probability $1 - 2\epsilon - 2\eta = 1 - O(\epsilon)$,
\[
 (f_0 \circ g^n)(Z) = (\h \circ (f_1,\ldots,f_m))(Z) = (\h \circ (f_1,\ldots,f_m))(W) = (f_0 \circ g^n)(W).
\]
(In more detail, if $f_m$ is $\epsilon$-close to the $T$-junta $F_m$, then $f_m(\col_m(Z)) = F_M(\col_m(Z)) = F_M(\col_m(W)) = f_m(\col_m(W))$ with probability $1 - 2\epsilon$.)

For each $x \in \{0,1\}^T$, define
\[
 \epsilon_0(x) = \Pr[(f_0 \circ g^n)(Z) \neq (f_0 \circ g^n)(W) \mid g(\row_i(Z)) = x_i \text{ for all } i \in T],
\]
so that $\E_{\mu_p}[\epsilon_0(x)] = O(\epsilon)$.

Recall that for $i \notin T$, $\row_i(W)$ is obtained from $\row_i(Z)$ by resampling the $m$th coordinate. Since $g$ depends on the $m$th coordinate, the joint distribution of $g(\row_i(Z)),g(\row_i(W))$ is $N_\rho^{\mu_p}$ for some $\rho < 1$. Therefore \Cref{lem:noise-stable-constant} shows that $(f_0)_{T \to x}$ is $O(\epsilon_0(x))$-close to a constant function with respect to $\mu_p$. In other words, $f_0$ is $O(\epsilon)$-close to a $T$-junta with respect to $\mu_p$.
\end{proof}

We now deduce that either one of $f_1,\ldots,f_m$ is close to constant, or all of $f_1,\ldots,f_m$ are close to juntas.

\begin{lemma} \label{lem:f-junta-multi-all}
Fix a function $g\colon \{0,1\}^m \to \{0,1\}$ which depends on all coordinates, and assume moreover that $g$ is not of the form $\bigoplus_{i=1}^m x_i \oplus b$. Let $p = \E[g]$.

For every $\epsilon > 0$ there exist $L \in \mathbb{N}$ and $\eta > 0$ such that if $f_0,\ldots,f_m$ is an $\eta$-approximate multi-polymorphism of $g$ then there exists a set $T \subseteq [n]$ of at most $L$ coordinates such that $f_0$ is $\epsilon$-close to a $T$-junta with respect to $\mu_p$, and furthermore one of the following cases holds:
\begin{enumerate}[(i)]
\item $f_1,\ldots,f_m$ are $\epsilon$-close to $T$-juntas.
\item One of $f_1,\ldots,f_m$ is $\epsilon$-close to a constant.
\end{enumerate}
\end{lemma}
\begin{proof}
Given $\epsilon > 0$, let $L,\eta_0$ be the parameters promised by \Cref{lem:f-junta-multi-f0} for $\epsilon^{2m}$, and let $\eta = \min(\eta_0,\epsilon^{2m})$. According to the lemma, $f_0$ is $\epsilon^{2m}$-close to a $T$-junta, where $T \subseteq [n]$ contains at most $L$ coordinates.

We can assume that all of $f_1,\ldots,f_m$ are $\epsilon$-far from constants, since otherwise case~(ii) holds. We will show that $f_m$ is $O(\epsilon^2)$-close to a $T$-junta; an identical argument works for $f_1,\ldots,f_{m-1}$.

We construct a pair $(Z,W)$ of coupled $n \times m$ zero-one matrices: we choose $Z$ uniformly at random, and let $W$ be the matrix obtained from $Z$ by resampling the entire $i$th row for all $i \notin T$.

With probability $1 - 2\epsilon^{2m} - 2\eta = 1 - O(\epsilon^{2m})$,
\[
 (\h \circ (f_1,\ldots,f_m))(Z) = (f_0 \circ g^n)(Z) = (f_0 \circ g^n)(W) = (\h \circ (f_1,\ldots,f_m))(W).
\]

Since $\h$ depends on all coordinates, there exist $\beta \in \{0,1\}^{m-1}$ and $b \in \{0,1\}$ such that $\h(\beta_1,\ldots,\beta_{m-1},x_m) = x_m \oplus b$. 

Let $j \in [m-1]$. Since $f_j$ is $\epsilon$-far from constant, $\Pr[f_j(\col_j(Z)) = \beta_j] \geq \epsilon$. We claim that moreover, $\Pr[f_j(\col_j(Z)) = f_j(\col_j(W)) = \beta_j] \geq \epsilon^2$. Indeed, for $x \in \{0,1\}^T$ define
\[
 \epsilon_j(x) = \Pr[f_j(\col_j(Z)) = \beta_j \mid \col_j(Z)|_T = x].
\]
Then $\E[\epsilon_j] \geq \epsilon$, and
\[
 \Pr[f_j(\col_j(Z)) = f_j(\col_j(W)) = \beta_j] = \E[\epsilon_j^2] \geq \E[\epsilon_j]^2 = \epsilon^2.
\]
It follows that with probability at least $\epsilon^{2(m-1)}$, we have $f_j(\col_j(Z)) = f_j(\col_j(W)) = \beta_j$ for all $j \in [m-1]$. Denote this event by $E$. Thus
\[
 \Pr[(\h \circ (f_1,\ldots,f_m))(Z) \neq (\h \circ (f_1,\ldots,f_m))(W) \mid E] = \frac{O(\epsilon^{2m})}{\epsilon^{2(m-1)}} = O(\epsilon^2).
\]
On the other hand, when $E$ happens,
\begin{align*}
 (\h \circ (f_1,\ldots,f_m))(Z) &= \h(\beta_1,\ldots,\beta_{m-1},f_m(\col_m(Z))) = f_m(\col_m(Z)) \oplus b, \\
 (\h \circ (f_1,\ldots,f_m))(W) &= \h(\beta_1,\ldots,\beta_{m-1},f_m(\col_m(W))) = f_m(\col_m(W)) \oplus b,
\end{align*}
and so
\[
 \Pr[f_m(\col_m(Z)) \neq f_m(\col_m(W)) \mid E] = O(\epsilon^2).
\]
Since $E$ only depends on the first $m-1$ columns, in fact
\[
 \Pr[f_m(\col_m(Z)) \neq f_m(\col_m(W))] = O(\epsilon^2).
\]

Finally, for $x \in \{0,1\}^T$ define
\[
 \epsilon_m(x) = \Pr[f_m(\col_m(Z)) \neq f_m(\col_m(W)) \mid \col_m(Z')|_T = x],
\]
so that $\E[\epsilon_m] = O(\epsilon^2)$. Since $\epsilon_m = 2\V{f_m|_{T \to x}}$, we conclude that $f_m$ is $O(\epsilon^2)$-close to a $T$-junta.
\end{proof}

We complete the proof of \Cref{lem:f-junta-multi} by handling case~(ii) of the preceding lemma.

\begin{proof}[Proof of \Cref{lem:f-junta-multi}]
Given $\epsilon > 0$, let $L,\eta_0$ be the parameters promised by \Cref{lem:f-junta-multi-all} for $\epsilon^m$, and let $\eta = \min(\eta_0,\epsilon^m)$.
Applying the lemma, we can assume that we are in case~(ii). Without loss of generality, suppose that $f_m$ is $\epsilon^m$-close to a constant.

We construct a pair $(Z,W)$ of coupled $n \times m$ zero-one matrices as follows: choose $Z$ uniformly at random, and let $W$ be the matrix obtained from $Z$ by resampling the $m$th column.

With probability $1 - 2\epsilon^m - 2\eta = 1 - O(\epsilon^m)$,
\[
 (f_0 \circ g^n)(Z) = (\h \circ (f_1,\ldots,f_m))(Z) = (\h \circ (f_1,\ldots,f_m))(W) = (f_0 \circ g^n)(W).
\]
Since $g$ depends on the $m$th coordinate, the joint distribution of the inputs to $f_0$ on both sides is $N_\rho^{\mu_p}$ for some $\rho < 1$. Therefore \Cref{lem:noise-stable-constant} shows that $f_0$ is $O(\epsilon^m)$-close to a constant $a_0$ with respect to $\mu_p$.

Let $J \subseteq [m]$ consist of the coordinates $j \in [m]$ such that $f_j$ is $\epsilon$-close to a constant $a_j$. Let $x \in \{0,1\}^m$ satisfy $x_j = a_j$ for all $j \in J$. For each $j \notin J$, since $f_j$ is $\epsilon$-far from constant, $f_j(\col_j(Z)) = x_j$ with probability at least $\epsilon$. For each $j \in J$, since $f_j$ is $\epsilon$-close to $a_j = x_j$, $f_j(\col_j(Z)) = x_j$ with probability at least $1-\epsilon$. Hence with probability at least $\epsilon^{m-|J|} (1-\epsilon)^{|J|} = \Omega(\epsilon^{m-1})$,
\[
 (\h \circ (f_1,\ldots,f_m))(Z) = \h(x).
\]
On the other hand, with probability $1 - \epsilon^m - \eta = 1- O(\epsilon^m)$,
\[
 (\h \circ (f_1,\ldots,f_m))(Z) = (f_0 \circ g^n)(Z) = a_0.
\]
Since $\Omega(\epsilon^{m-1}) > O(\epsilon^m)$, both events hold simultaneously with positive probability, and so $\h(x) = a_0$, as desired.
\end{proof}

\section{List decoding regime}
\label{sec:list-decoding}

A polymorphism of $g$ is a function $f$ satisfying $f \circ g^n = g \circ f^m$. There are two ways to relax this definition:

\begin{itemize}
    \item The 99\% regime: study functions $f$ satisfying $f \circ g^n = g \circ f^m$ for most inputs. \Cref{thm:main} shows that such functions are close to exact polymorphisms.
    \item The 1\% regime: study functions $f$ satisfying $f \circ g^n = g \circ f^m$ with \emph{significant} probability. We would like to say that such functions are \emph{structured}. 
\end{itemize}

When $g$ is the XOR function, the classic analysis of linearity testing~\cite{Bellare} shows that if $\Pr[f(x \oplus y) = f(x) \oplus f(y)] \geq 1/2 + \epsilon$, then $f$ is correlated with some character, that is, for some $S \subseteq [n]$,
\[
 \Pr\left[f(x) = \bigoplus_{i \in S} x_i \right] \geq \frac{1}{2} + \epsilon.
\]
Conversely, if $f$ is a random function then $\Pr[f(x \oplus y) = f(x) \oplus f(y)] \approx 1/2$, showing that $1/2$ is the correct threshold for this kind of structure.

What happens for other $g$? Let us take the AND function as a test case. If we choose $f$ at random then $\Pr[f(x \land y) = f(x) \land f(y)] \approx 1/2$, and so one could conjecture that when $\Pr[f(x \land y) = f(x) \land f(y)] \geq 1/2 + \epsilon$ then $f$ is correlated with some character. The Majority function refutes this conjecture, since it satisfies $\Pr[f(x \land y) = f(x) \land f(y)] \approx 3/4$ but is not correlated with any character.

The threshold $3/4$ is natural, since it corresponds to the following ``almost-random'' construction: choose $f$ at random for inputs whose Hamming weight is close to $n/2$ (where $n$ is the input size), and choose $f$ to be~$0$ elsewhere. However, it is not the correct threshold: if we choose $f$ to be Majority for inputs whose Hamming weight is close to $n/2$, and a \emph{biased} majority for other inputs, then $\Pr[f(x \land y) = f(x) \land f(y)] \approx 0.814975$.

Our main result shows that $0.814975$ is the correct threshold for AND: if $\Pr[f(x \land y) = f(x) \land f(y)] \geq 0.814975 + \epsilon$ then $f$ is correlated with some character. Moreover, we can guarantee that this character has low degree.

The idea of the proof is to translate the question about Boolean variables to a question on Gaussian space. To this end, we define a Gaussian analog of the distribution $(g(x),x)$. It will be convenient to switch from $\{0,1\}$ to $\{-1,1\}$.

\begin{definition} \label{def:gaussian-g}
Let $g\colon \{-1,1\}^m \to \{-1,1\}$ be non-constant. The distribution $\normalg{g}$ is an $(m+1)$-variate Gaussian distribution $(\cG_0,\cG_1,\ldots,\cG_m)$ given by:
\begin{itemize}
    \item Each coordinate is a standard Gaussian.
    \item The Gaussians $\cG_1,\ldots,\cG_m$ are independent.
    \item For each $j \in [m]$,
    \[
     \E[\cG_0 \cG_j] = \frac{\hat{g}(\{j\})}{\sqrt{1 - \hat{g}(\emptyset)^2}}.
    \]
\end{itemize}
\end{definition}

Our main result states that if $\Pr[f \circ g^n = g \circ f^m]$ exceeds a certain threshold, then $f$ is correlated with a low-degree character. The result applies to any function other than XOR or NXOR; analogous results for these functions (in which the character need not be low-degree) follow from a generalization of the arguments in~\cite{Bellare}.

\begin{definition} \label{def:sg}
Fix a function $g\colon \{-1,1\}^m \to \{-1,1\}$. Let $s_g$ be the infimum over all $s$ for which the following holds.

For every $\epsilon > 0$ there exist $\delta > 0$ and $L \in \mathbb{N}$ such that for all $n \in \mathbb{N}$, if $f\colon \{-1,1\}^n \to \{-1,1\}$ satisfies
\[
 \Pr[f \circ g^m = g \circ f^m] \geq s + \epsilon
\]
then $f$ has correlation at least $\delta$ with some character of degree at most $L$, that is, there exists $S \subseteq [n]$, of size at most $L$, such that
\[
 \left| \E\left[ f(x) \prod_{i \in S} x_i \right] \right| \geq \delta.
\]
\end{definition}

\begin{theorem} \label{thm:list-decoding}
Fix a function $g\colon \{-1,1\}^m \to \{-1,1\}$ which depends on all coordinates and is not $\pm \prod_{i=1}^m x_i$.

If $\E[g] \neq 0$ then let $s_g^U$ be the supremum of
\[
 \frac{1}{2} + \frac{1}{2}
 \E_{x \sim \stdnormal^n}
 \left[\Bigl|
 \E_{\cG \sim \normalg{g}^n}[g(q_1(\cG_1), \ldots, q_m(\cG_m)) \mid \cG_0 = x]
 \Bigr|\right]
\]
over all $n \in \mathbb{N}$ and all functions $q_1,\ldots,q_m\colon \mathbb{R}^n \to \{-1,1\}$ satisfying $\E[q_1] = \cdots = \E[q_m] = 0$.

If $\E[g] = 0$, instead let $s_g^U$ be the supremum of
\[
 \frac{1}{2} + \frac{1}{2}
 \E_{\cG \sim \normalg{g}^n}[q_0(\cG_0) g(q_1(\cG_1), \ldots, q_m(\cG_m))]
\]
over all $n \in \mathbb{N}$ and all functions $q_0,\ldots,q_m\colon \mathbb{R}^n \to \{-1,1\}$ satisfying $\E[q_0] = \cdots = \E[q_m] = 0$.

Then $s_g \leq s_g^U$.
\end{theorem}

We can show that $s_g^U < 1$ for all functions $g$ covered by the theorem.

\begin{lemma} \label{lem:sg-ub}
If $g\colon \{-1,1\}^m \to \{-1,1\}$ depends on all coordinates and is not $\pm \prod_{i=1}^m x_i$ then $s_g^U < 1$.
\end{lemma}

If we take the supremum in \Cref{thm:list-decoding} with the additional constraint that the functions $q_1,\ldots,q_m$ coincide, then the resulting value is a lower bound on $s_g$.

\begin{lemma} \label{lem:sg-lb}
Fix a function $g\colon \{-1,1\}^m \to \{-1,1\}$ which depends on all coordinates and is not $\pm \prod_{i=1}^m x_i$.

If $\E[g] \neq 0$ then let $s_g^L$ be the supremum of
\[
 \frac{1}{2} + \frac{1}{2}
 \E_{x \sim \stdnormal^n}
 \left[\Bigl|
 \E_{\cG \sim \normalg{g}^n}[g(q(\cG_1), \ldots, q(\cG_m)) \mid \cG_0 = x]
 \Bigr|\right]
\]
over all $n \in \mathbb{N}$ and all functions $q\colon \mathbb{R}^n \to [-1,1]$ satisfying $\E[q] = 0$.

If $\E[g] = 0$, instead let $s_g^L$ be the supremum of
\[
 \frac{1}{2} + \frac{1}{2}
 \E_{\cG \sim \normalg{g}^n}[q(\cG_0) g(q(\cG_1), \ldots, q(\cG_m))]
\]
over all $n \in \mathbb{N}$ and all functions $q\colon \mathbb{R}^n \to [-1,1]$ satisfying $\E[q] = 0$.

There exists a sequence of functions $f_N\colon \{-1,1\}^N \to \{-1,1\}$, with $N \to \infty$, such that
\[
 \Pr[f_N \circ g^N = g \circ f_N^m] \longrightarrow s^L_g,
\]
and for each $\delta > 0$ and $L \in \mathbb{N}$, for large enough $N$ the functions $f_N$ do not have correlation at least $\delta$ with any character of degree at most $L$.

Consequently, $s_g \geq s_g^L$.
\end{lemma}

\noindent We do not know whether $s_g^U > s_g^L$ holds for some $g$, that is, whether there is any advantage in allowing $q_1,\ldots,q_m$ to be different.

When $\E[g] \neq 0$, for any function $q$ satisfying $\E[q] = 0$ we have
\[
 \E_{x \sim \stdnormal^n}
 \left[\Bigl|
 \E_{\cG \sim \normalg{g}^n}[g(q(\cG_1), \ldots, q(\cG_m)) \mid \cG_0 = x]
 \Bigr|\right] \geq
 \Bigl|\E_{\cG \sim \normalg{g}^n}[g(q(\cG_1), \ldots, q(\cG_m))]\Bigr| = |\E[g]|,
\]
and so
\[
 s_g^L \geq \frac{1}{2} + \frac{1}{2} |\E[g]|.
\]
This corresponds to the trivial construction in which $f$ is chosen randomly around the middle slice, and $\sgn(\E[g])$ elsewhere. The following result, proved by taking $n = 1$ and $q = \sgn$, shows that we can improve it if all Fourier coefficients of $g$ on the first level are non-zero.

\begin{lemma} \label{lem:sg-nontrivial}
If $g\colon \{-1,1\}^m \to \{-1,1\}$ satisfies $\E[g] \neq 0$ and $\hat{g}(\{x_j\}) \neq 0$ for all $j \in [m]$ then $s_g^L > \frac{1}{2} + \frac{1}{2} |\E[g]|$.
\end{lemma}

When $g$ is an AND function,  we can show that $s_g^U$ is attained by $n = 1$ and $q_1(x) = \cdots = q_m(x) = \sgn(x)$. In view of \Cref{lem:sg-lb}, this gives an expression for $s_g$.

\begin{theorem} \label{thm:sg-AND}
Let $m \geq 2$, and let $g(x_1,\ldots,x_m) = \min(x_1,\ldots,x_m)$. Then
\[
 s_g = 
 \frac{1}{2} + \frac{1}{2}
 \E_{x \sim \stdnormal}
 \left[\Bigl|
 \E_{\cG \sim \normalg{g}}[g(\sgn(\cG_1), \ldots, \sgn(\cG_m)) \mid \cG_0 = x]
 \Bigr|\right].
\]

In particular, when $m = 2$, we obtain
\[
 s_\land \approx 0.814975356673002,
\]
where $\land$ denotes the binary AND function $g(a,b) = a \land b = \min(a,b)$.
\end{theorem}

\begin{remark}
If $g\colon \{-1,1\}^m \to \{-1,1\}$ is of the form $\pm \prod_{i=1}^m x_i$ for $m \geq 2$ and $\Pr[f \circ g^n = g \circ f^m] \geq \frac{1}{2} + \epsilon$, then the proof of \Cref{lem:approximate-polymorphism-xor} shows that $f$ has correlation $\Omega(\epsilon^{1/(m-1)})$ with \emph{some} character. In contrast to \Cref{thm:list-decoding}, we cannot guarantee correlation with a low-degree character.
\end{remark}

\subsection{Preliminaries}

For $\lambda \in (-1,1)$ and (implicit) $n$, we denote by $\pi_\lambda$ the product distribution on $\{-1,1\}^n$ in which each coordinate has expectation $\lambda$ (this is the analog of the distributions $\mu_p$).
We denote the corresponding noise operator by $T_\rho^{\lambda}$; it multiplies the $d$th level of the $\lambda$-biased Fourier expansion by $\rho^d$. If $\lambda = 0$ (corresponding to the uniform distribution) then we sometimes omit the superscript.

We extend functions on $\{-1,1\}^n$ to functions on $\mathbb{R}^n$ multilinearly, that is, via the Fourier expansion (it doesn't matter which bias we use to define the Fourier expansion). The analog of $T_\rho^\lambda$ in Gaussian space is $U_\rho^\lambda$.

The function $\clip\colon \mathbb{R} \to [-1,1]$ clips its argument to $[-1,1]$:
\[
 \clip(x) = \begin{cases} -1 & \text{if } x < -1, \\ x & \text{if } -1 \leq x \leq 1 \\ 1 & \text{if } x > 1. \end{cases}
\]

\subsubsection{Invariance principle}

Our proof will require a multidimensional version of the invariance principle due to Mossel~\cite{Mossel2010}.

\begin{theorem} \label{thm:invariance}
Let $\mathcal{D}$ be an arbitrary distribution on $\{-1,1\}^k$, and let $\mathcal{G}$ be a multivariate Gaussian distribution with the same mean vector and covariance matrix. Denote the expectations of the individual coordinates by $\mu_1,\ldots,\mu_k$.

For every $\epsilon,\gamma > 0$ there exist $d \in \mathbb{N}$ and $\tau > 0$ such that for any $f_1,\ldots,f_k\colon \{-1,1\}^n \to \{-1,1\}$ and $\gamma_1,\ldots,\gamma_k \geq \gamma$, if $\Inf_i^{\leq d}[f_j] \leq \tau$ for all $j \in [k]$ (where influence is with respect to $\pi_{\mu_j}$) and $i \in [n]$ then
\[
 \left|
 \E_{x \sim \mathcal{D}^n}
 [(T_{1-\gamma_1}^{\mu_1} f_1)(x_1) \cdots (T_{1-\gamma_k}^{\mu_k} f_k)(x_k)]
 -
 \E_{g \sim \mathcal{G}^n}
 [\clip(T_{1-\gamma_1}^{\mu_1} f_1)(g_1) \cdots \clip(T_{1-\gamma_k}^{\mu_k} f_k)(g_k)]
 \right| < \epsilon.
\]

Furthermore, for each $j \in [k]$,
\[
 |\E_{g \sim \mathcal{G}^n}[\clip(T_{1-\gamma_j}^{\mu_j} f_j)(g_j)] - \E_{x \in \mathcal{D}^n}[f_j(x_j)]| < \epsilon.
\]
\end{theorem}
\begin{proof}
\cite[Theorem 4.1]{Mossel2010} shows that for appropriate $d,\tau$ we have
\[
 \left|
 \E_{x \sim \mathcal{D}^n}
 [(T_{1-\gamma_1}^{\mu_1} f_1)(x_1) \cdots (T_{1-\gamma_k}^{\mu_k} f_k)(x_k)]
 -
 \E_{g \sim \mathcal{G}^n}
 [(T_{1-\gamma_1}^{\mu_1} f_1)(g_1) \cdots (T_{1-\gamma_k}^{\mu_k} f_k)(g_k)]
 \right| < \frac{\epsilon}{2}.
\]

\cite[Theorem 4.2]{Mossel2010} shows that for appropriate $d,\tau$, for every $j \in [k]$ we have
\[
 \left|
 \E_{g \sim \mathcal{G}^n} [((T_{1-\gamma_j}^{\mu_j} f_j)(g_j) - \clip(T_{1-\gamma_j}^{\mu_j} f_j)(g_j))^2]
 \right| < \left(\frac{\epsilon}{2k}\right)^2.
\]
Applying Cauchy--Schwarz $k$ times implies that
\[
 \left|
 \E_{g \sim \mathcal{G}^n}
 [(T_{1-\gamma_1}^{\mu_1} f_1)(g_1) \cdots (T_{1-\gamma_k}^{\mu_k} f_k)(g_k)]
 -
 \E_{g \sim \mathcal{G}^n}
 [\clip(T_{1-\gamma_1}^{\mu_1} f_1)(g_1) \cdots \clip(T_{1-\gamma_k}^{\mu_k} f_k)(g_k)]
 \right| < \frac{\epsilon}{2},
\]
completing the proof of the main estimate.

The bound given by \cite[Theorem 4.2]{Mossel2010} also implies that for every $j \in [k]$,
\[
 \left|\E_{g \sim \mathcal{G}^n}[\clip(T_{1-\gamma_j}^{\mu_j} f_j)(g_j)] - \E_{g \in \mathcal{G}^n}[(T_{1-\gamma_j}^{\mu_j} f_j)(g_j)]\right| < \epsilon.
\]
Since $\E[(T_{1-\gamma_j}^{\mu_j} f_j)(g_j)] = \E[f_j(g_j)] = \E[f_j(x_j)]$, this completes the proof.
\end{proof}

We will apply this theorem to the distribution $(g(x),x)$, whose Gaussian analog we now calculate.

\begin{lemma} \label{lem:normalg}
Let $g\colon \{-1,1\}^m \to \{-1,1\}$ be a non-constant function, and let $\mu = \E[g]$. If $x$ is distributed uniformly over $\{-1,1\}^m$ then the vector
\[
 \frac{g(x) - \mu}{\sqrt{1 - \mu^2}},x_1,\ldots,x_m
\]
has the same mean vector and covariance matrix as $\normalg{g}$.
\end{lemma}
\begin{proof}
By construction, each coordinate has zero mean and unit variance, and the last $m$ coordinates are independent. Hence it remains to compute the covariance between the first coordinate and the $(j+1)$th coordinate, which corresponds to $x_j$:
\[
 \E_{x \in \{-1,1\}^m} \left[ \frac{g(x) - \mu}{\sqrt{1 - \mu^2}} x_j \right] = \frac{\hat{g}(\{j\})}{\sqrt{1 - \mu^2}}. \qedhere
\]
\end{proof}

\subsubsection{Connected distributions}

\Cref{thm:invariance} requires adding noise to the function. We will do this in the context of evaluating
\[
 \E_Z[(f \circ g^n)(Z) \cdot (g \circ f^m)(Z)].
\]
In order to show that the noise only slightly changes this expression, we will use a result of Mossel~\cite{Mossel2010} on connected distributions.

\begin{definition}[Connected distribution]
\label{def:connected-space}
Let $\mathcal{D}$ be a distribution on $\{-1,1\}^k$. For $i \in \{1,\ldots,k-1\}$, let $G_i$ be the bipartite graph $(X,Y,E)$, where
\begin{itemize}
    \item $X$ is the projection of the support of $\mathcal{D}$ to the first $i$ coordinates.
    \item $Y$ is the projection of the support of $\mathcal{D}$ to the last $k-i$ coordinates.
    \item $(x,y) \in E$ if $(x,y)$ is in the support of $\mathcal{D}$.
\end{itemize}
The distribution $\mathcal{D}$ is \emph{connected} if all of $G_1,\ldots,G_{k-1}$ are connected.
\end{definition}

\begin{theorem} \label{thm:connected}
Fix a connected distribution $\mathcal{D}$ on $\{-1,1\}^k$ with mean vector $\mu$. For every $\epsilon > 0$ there are $\gamma_1,\ldots,\gamma_k > 0$ such that the following holds.

If $f_1,\ldots,f_k\colon \{-1,1\}^n \to [-1,1]$ then
\[
 \E_{x \sim \mathcal{D}^n}[f_1(x_1) \cdots f_k(x_k) - (T_{1-\gamma_1}^{\mu_1} f_1)(x_1) \cdots (T_{1-\gamma_k}^{\mu_k} f_k)(x_k)] < \epsilon.
\]
\end{theorem}

We prove this in \Cref{apx:connected}, using ideas from Mossel~\cite{Mossel2010}. 

We are interested in the distribution $(g(x),x)$. While this distribution is not connected, we can show that unless $g$ is XOR or NXOR, we can rearrange it so that it becomes connected.

\begin{lemma} \label{lem:g-connected}
Fix a function $g\colon \{-1,1\}^m \to \{-1,1\}$ which depends on all coordinates and is not $\pm \prod_{i=1}^m x_i$.

After a suitable permutation, the distribution $(g(x),x)$ (where $x$ is chosen uniformly from $\{-1,1\}^m$) is connected.
\end{lemma}
\begin{proof}
By factoring the multilinear representation of $g$, we can uniquely decompose $g$ as
\[
 g(x) = g_1(x|_{S_1}) \cdots g_\ell(x|_{S_\ell}),
\]
where $S_1,\ldots,S_\ell$ form a decomposition of $[m]$, and none of $g_1,\ldots,g_\ell$ can be further decomposed in this way. Since $g$ is not $\pm \prod_{i=1}^m x_i$, one of the sets, say $S_1$, contains at least two points. We rearrange the inputs to $g$ so that $1,m \in S_1$. We will prove that
\[
 x_1,x_2,\ldots,x_{m-1},g(x),x_m
\]
is connected.

We will show that the graph $G_i$ is connected for $i \in [m-1]$; for $i = m$, the same argument works, with $X$ and $Y$ switching roles.

Let $i \in [m-1]$.
By construction, $X = \{-1,1\}^i$, and (up to permutation) $Y$ is the collection of all $(y,b) \in \{-1,1\}^{(m-i)+1}$ such that $g(x,y) = b$ for some $x \in X$. The edges are $(x,(y,g(x,y))$.

Since $g$ depends on all coordinates, there exists some $y \in \{-1,1\}^{m-i}$ such that $(y,1),(y,-1) \in Y$. Let $X_- = \{ x \in X : g(x,y) = -1 \}$ and $X_+ = \{ x \in X : g(x,y) = 1 \}$; note that $X_-,X_+ \neq \emptyset$. We claim that there exist $x_- \in X_-$, $x_+ \in X_+$, and $(z,b) \in Y$ such that $g(x_-,z) = g(x_+,z)$. Otherwise, for every $x_- \in X_-$ and $x_+ \in X_+$ we have $g(x_-,z) g(x_+,z) = -1$. This implies that $g(x,z)$ only depends on whether $x \in X_+$ or $x \in X_-$, and so we can write $g(x,z) = g_x(x) g_z(z)$, which is impossible since $g_x$ depends on $x_1$ and $g_z$ depends on $x_m$, and we assumed that $g_1(S_1)$ is indecomposable.

Recall $b = g(x_-,z) = g(x_+,z)$. The path
\[
 (y,1) - x_+ - (z,b) - x_- - (y,-1)
\]
shows that $(y,1)$ and $(y,-1)$ are connected.

If $(z,b) \in Y$ is arbitrary, let $x \in X$ be such that $g(x,z) = b$. Then $(z,b)$ is connected to $(y,g(x,y))$ via $x$. Therefore all vertices in $Y$ are connected. Since every vertex in $X$ is connected to some vertex in $Y$, we conclude that the entire graph is connected.
\end{proof}

\subsubsection{Borell's theorem}

The upper bound in \Cref{lem:sg-ub} derives from Borell's isoperimetric theorem~\cite{Borell} in Gaussian space, in the following form.

\begin{theorem} \label{thm:borell}
Let $f\colon \mathbb{R}^n \to [-1,1]$ satisfy $\E[f] = 0$ with respect to the standard Gaussian measure. For every $\rho \in [-1,1]$,
\[
 \E[(U_\rho f)^2] \leq \frac{2}{\pi} \arcsin (\rho^2),
\]
which is tight for $f(x) = \sgn(x_1)$.
\end{theorem}

The upper bound in \Cref{thm:sg-AND} derives from a generalization of Borell's theorem due to Neeman~\cite{Neeman}.

\begin{theorem} \label{thm:neeman}
Let $f_1,\ldots,f_m\colon \mathbb{R}^n \to \{0,1\}$ be arbitrary functions. Let $\mathcal{G} = (\mathcal{G}_1,\ldots,\mathcal{G}_m) \sim \normal(\mu,\Sigma)^n$ be an $m$-dimensional multivariate Gaussian distribution, where all elements of $\Sigma$ are non-negative.

Let $F_1,\ldots,F_m\colon \mathbb{R}^n \to \{0,1\}$ be functions of the form $F_i(x_1,\ldots,x_n) = 1_{x_1 \leq \theta_i}$, where $\theta_1,\ldots,\theta_m$ are chosen so that $\E[F_i] = \E[f_i]$ for all $i \in [m]$. Then
\[
 \E[f_1(\mathcal{G}_1) \cdots f_m(\mathcal{G}_m)] \leq
 \E[F_1(\mathcal{G}_1) \cdots F_m(\mathcal{G}_m)].
\]
\end{theorem}

\subsection{Proof of \crtCref{thm:list-decoding}} 

We divide the proof into two parts. In the first part, we show that \Cref{thm:list-decoding} holds when the functions $q_0,\ldots,q_m$ in the definition of $s_g^U$ are allowed to be $[-1,1]$-valued; we denote the resulting value by $S_g^U$. We then show that the value of $s_g^U$ stays the same even if we restrict $q_0,\ldots,q_m$ to be $\{-1,1\}$-valued.

\subsubsection{Proof for fractional functions}

Let us first notice that if $\E[g] \neq 0$ then $S_g^U$ is the supremum of
\[
 \frac{1}{2} + \frac{1}{2} \E_{\cG \sim \normalg{g}^n}[q_0(\cG_0) g(q_1(\cG_1), \ldots, q_m(\cG_m))]
\]
over all $n \in \mathbb{N}$ and all functions $q_0,q_1,\ldots,q_m\colon \mathbb{R}^n \to [-1,1]$ satisfying $\E[q_1] = \cdots = \E[q_m] = 0$.
When $\E[g] = 0$, we have defined $S_g^U$ using a similar expression, with the additional constraint $\E[q_0] = 0$.

Let $\|\hat{g}\|_1 = \sum_{S \subseteq [m]} |\hat{g}(S)|$, and set
\[
 \epsilon_1 = \frac{\epsilon}{3\|\hat{g}\|_1}, \quad
 \epsilon_2 = \frac{\epsilon}{(6m+3)\|\hat{g}\|_1}, \quad
 \epsilon_3 = \frac{\epsilon}{6m\|\hat{g}\|_1}, \quad
 \epsilon_4 = \frac{\epsilon}{m+1}.
\]
Let $\gamma_0,\ldots,\gamma_m > 0$ be the parameters given by \Cref{thm:connected} for $\epsilon_1$. Let $d,\tau$ be the parameters given by \Cref{thm:invariance} for $\epsilon_2$ and $\gamma = \min(\gamma_0,\ldots,\gamma_m)$. Let $L$ be the parameter given by \Cref{thm:jones} for $d,\tau,\epsilon_4$. Finally, we define $\delta = 2^{-L} \epsilon_3$. We will assume that $|\hat{f}(S)| \leq \delta$ for all $|S| \leq L$, and show that $\Pr[f \circ g^n = g \circ f^m] \leq s^U_g + \epsilon$.

According to \Cref{thm:jones}, we can find a set $T$ of at most $L$ coordinates such that with probability $1-\epsilon_4$, $f_{T \to z}$ is $(d,\tau)$-regular, and this holds both with respect to the uniform probability and with respect to $\pi_{\E[g]}$. In particular, if we choose $W \sim \{-1,1\}^{T \times m}$ then with probability $1 - (m+1)\epsilon_4 = 1 - \epsilon$, the functions $f_j = f_{T \to \col_j(W)}$ are $(d,\tau)$-regular with respect to $\pi_0$ for all $j \in [m]$, and the function $f_0 = f_{T \to g(W)}$ (where $g$ is applied row by row) is $(d,\tau)$-regular with respect to $\pi_{\E[g]}$. If these properties are satisfied, we say that $W$ is \emph{good}.

For any $W$ and any $j \in [m]$, the expected value of $f_j$ is
\[
 \sum_{S \subseteq T} \hat{f}(S) \prod_{i \in S} W_{ij},
\]
which is at most $2^L \delta = \epsilon_3$ in absolute value.

Let $W$ be a good partial input, and denote by $\mathcal{D}(W)$ all inputs compatible with $W$. Let $q_j = \clip(T_{1-\gamma_j} f_j)$ for $j \in [m]$, and let $q_0(g) = \clip(T_{1-\gamma_j}^{\E[g]} f_0)(\sqrt{1-\E[g]^2} \cdot g + \E[g])$ (we define $q_0$ so that if its input is a standard Gaussian, then the input to $f_0$ has the same mean and variance as $g(x)$). Combining \Cref{thm:invariance} and \Cref{thm:connected} shows that for each $S \subseteq [m]$,
\[
 \left|
 \E_{Z \in \mathcal{D}(W)}
 \left[(f \circ g^n)(Z) \prod_{j \in S} f_j(\col_j(Z))\right]
 -
 \E_{\mathcal{G} \sim \normalg{g}^n}
 \left[q_0(\mathcal{G}_0) \prod_{j \in S} q_j(\mathcal{G}_j)\right]
 \right| \leq \epsilon_1 + \epsilon_2.
\]
\Cref{thm:invariance} also shows that for all $j \in [m]$, $|\E[q_j]| \leq |\E[f_j]| + \epsilon_2 \leq \epsilon_2 + \epsilon_3$. Let \[ q'_j = \frac{q_j - \E[q_j]}{1 + |\E[q_j]|}, \] so that $\E[q'_j] = 0$. Then
\[
 \left|
 \E_{\mathcal{G} \sim \normalg{g}^n}
 \left[q_0(\mathcal{G}_0) \prod_{j \in S} q_j(\mathcal{G}_j)\right] -
 \E_{\mathcal{G} \sim \normalg{g}^n}
 \left[q_0(\mathcal{G}_0) \prod_{j \in S} \frac{q_j(\mathcal{G}_j)}{1 + |\E[q_j]|}\right]
 \right| \leq
 1 -
 \prod_{j \in S} \frac{1}{1 + |\E[q_j]|} \leq m(\epsilon_2 + \epsilon_3),
\]
and
\[
 \left|
  \E_{\mathcal{G} \sim \normalg{g}^n}
 \left[q_0(\mathcal{G}_0) \prod_{j \in S} \frac{q_j(\mathcal{G}_j)}{1 + |\E[q_j]|}\right] -
 \E_{\mathcal{G} \sim \normalg{g}^n}
 \left[q_0(\mathcal{G}_0) \prod_{j \in S} q'_j(\mathcal{G}_j)\right]
 \right| \leq
 \sum_{j \in S} \frac{|\E[q_j]|}{1 + |\E[q_j]|} \leq m(\epsilon_2 + \epsilon_3).
\]
Combining the three estimates,
multiplying by $\hat{g}(S)$, and summing over all $S \subseteq [m]$, we conclude that
\[
 \left|
 \E_{Z \in \mathcal{D}(W)}[(f \circ g^n)(Z) (g \circ f^m)(Z)] -
 \E_{\mathcal{G} \sim \normalg{g}^n}[q_0(\mathcal{G}_0) g(q'_1(\mathcal{G}_1), \ldots, q'_m(\mathcal{G}_m))]
 \right| \leq \|\hat{g}\|_1 (\epsilon_1 + (2m+1)\epsilon_2 + 2m\epsilon_3) = \epsilon.
\]
When $\E[g] = 0$, we similarly replace $q_0$ with $q'_0$, adjusting $\epsilon_2,\epsilon_3$ accordingly.
We conclude that
\[
 \E_{Z \in \mathcal{D}(W)}[(f \circ g^n)(Z) (g \circ f^m)(Z)] \leq (2S_g^U - 1) + \epsilon.
\]

Considering all possible $W$, this shows that
\[
 \E_Z[(f \circ g^n)(Z) (g \circ f^m)(Z)] \leq (2S_g^U - 1) + 2\epsilon.
\]
It follows that
\[
 \Pr[f \circ g^n = g \circ f^m] \leq S_g^U + \epsilon.
\]

\subsubsection{Rounding}

Let $h(x_0,\ldots,x_m) = x_0 g(x_1,\ldots,x_m)$, and note that $h$ is multilinear.

To complete the proof of \Cref{thm:list-decoding}, we show that $S_g^U = s_g^U$. Recall that $2S_g^U-1$ is the supremum of
\[
 \E_{\mathcal{G} \sim \normalg{g}^n}[h(q_0(\mathcal{G}_0),\ldots,q_m(\mathcal{G}_m))],
\]
where $n \in \mathbb{N}$ and $q_0,\ldots,q_m\colon \mathbb{R}^n \to [-1,1]$ satisfy $\E[q_1] = \cdots \E[q_m] = 0$; when $\E[g] = 0$, we also require $\E[q_0] = 0$. In contrast, $2s_g^U - 1$ is defined similarly, with the functions $q_0,\ldots,q_m$ being $\{-1,1\}$-valued.

For every $\epsilon > 0$ we can find $n \in \mathbb{N}$ and functions $q_0,\ldots,q_m\colon \mathbb{R}^n \to [-1,1]$ satisfying $\E[q_1] = \cdots = \E[q_m] = 0$ (and $\E[q_0] = 0$ if $\E[g] = 0$), and furthermore
\[
 \E_{\mathcal{G} \sim \normalg{g}^n}[h(q_0(\mathcal{G}_0),\ldots,q_m(\mathcal{G}_m))] \geq 2S_g^U - 1 - \epsilon.
\]
We will show that for every $\delta > 0$, we can replace each $q_i$ by a $\{-1,1\}$-valued function $Q_i$, with the same expectation as $q_i$, such that
\[
 \E_{\mathcal{G} \sim \normalg{g}^n}[h(Q_0(\mathcal{G}_0),\ldots,Q_m(\mathcal{G}_m))] \geq 2S_g^U - 1 - \epsilon - O(\delta).
\]

We replace the functions one by one. At step $i$, given $Q_0,\ldots,Q_{i-1},q_i,\ldots,q_m$, we find a function $Q_i$ such that
\begin{multline*}
 \E_{\mathcal{G} \sim \normalg{g}^n}[h(Q_0(\mathcal{G}_0),\ldots,Q_{i-1}(\mathcal{G}_{i-1}),Q_i(\mathcal{G}_i),q_{i+1}(\mathcal{G}_{i+1}),\ldots,q_m(\mathcal{G}_m))] \geq \\
 \E_{\mathcal{G} \sim \normalg{g}^n}[h(Q_0(\mathcal{G}_0),\ldots,Q_{i-1}(\mathcal{G}_{i-1}),q_i(\mathcal{G}_i),q_{i+1}(\mathcal{G}_{i+1}),\ldots,q_m(\mathcal{G}_m))] - O(\delta).
\end{multline*}
Since $q_i$ is bounded, we can discretize it to a function $\tilde{q}_i$ which attains finitely many values such that $|q_i(x) - \tilde{q}_i(x)| \leq \delta$ for all $x \in \mathbb{R}$, and furthermore $\E[\tilde{q}_i] = \E[q_i]$. One way to do so is to divide $[-1,1]$ into $2/\delta$ intervals of width $\delta$, and for each such interval $I$, replace $q_i^{-1}(I)$ with its expectation.
Since $Q_0,\ldots,Q_{i-1},q_{i+1},\ldots,q_m$ are all bounded,
\begin{multline*}
 \E_{\mathcal{G} \sim \normalg{g}^n}[h(Q_0(\mathcal{G}_0),\ldots,Q_{i-1}(\mathcal{G}_{i-1}),\tilde{q}_i(\mathcal{G}_i),q_{i+1}(\mathcal{G}_{i+1}),\ldots,q_m(\mathcal{G}_m))] \geq \\
 \E_{\mathcal{G} \sim \normalg{g}^n}[h(Q_0(\mathcal{G}_0),\ldots,Q_{i-1}(\mathcal{G}_{i-1}),q_i(\mathcal{G}_i),q_{i+1}(\mathcal{G}_{i+1}),\ldots,q_m(\mathcal{G}_m))] - O(\delta).
\end{multline*}

The idea now is to find a distribution $\mathbf{Q}_i$ of $\{-1,1\}$-valued functions such that $\E[\mathbf{Q}_i(x)] = \tilde{q}_i(x)$ for all $x \in \mathbb{R}$, and furthermore each distribution in the support of $\mathbf{Q}_i$ has expectation $\E[q_i]$. Since the function $h$ is multilinear,
\begin{multline*}
 \E_{\mathbf{Q}_i}
 \E_{\mathcal{G} \sim \normalg{g}^n}[h(Q_0(\mathcal{G}_0),\ldots,Q_{i-1}(\mathcal{G}_{i-1}),\mathbf{Q}_i(\mathcal{G}_i),q_{i+1}(\mathcal{G}_{i+1}),\ldots,q_m(\mathcal{G}_m))] = \\
 \E_{\mathcal{G} \sim \normalg{g}^n}[h(Q_0(\mathcal{G}_0),\ldots,Q_{i-1}(\mathcal{G}_{i-1}),\tilde{q}_i(\mathcal{G}_i),q_{i+1}(\mathcal{G}_{i+1}),\ldots,q_m(\mathcal{G}_m))]
\end{multline*}
Therefore some function $Q_i$ in the support of $\mathbf{Q}_i$ satisfies our requirements.

It remains to construct the distribution $\mathbf{Q}_i$. For each value $c$ in the support of $\tilde{q}_i$, consider $\tilde{q}_i^{-1}(c)$, and put it in a measure-preserving bijection $\alpha$ with some interval $[0,1]$ (with respect to an appropriately scaled Lebesgue measure). We round $\tilde{q}_i^{-1}(c)$ as follows: choose $\theta \in [0,1]$ at random, give $\alpha^{-1}([\theta,\theta+\frac{1+c}{2}])$ (wrapping if needed) the value $+1$, and the rest of $\tilde{q}_i^{-1}(c)$ the value $-1$. This guarantees the required properties.

\subsection{Proof of \crtCref{lem:sg-ub}}

For $j \in [m]$, define
\[
 \rho_j = \E_{\cG \sim \normalg{g}}[\cG_0 \cG_j] = \frac{\hat{g}(\{j\})}{\sqrt{1 - \E[g]^2}},
\]
and let
\[
 a = \sqrt{\sum_{j=1}^m \rho_j^2} = \sqrt{\frac{\|g^{=1}\|^2}{\|g^{\ge 1}\|^2}} < 1.
\]

Let $Z_1,\ldots,Z_m,W_1,\ldots,W_m$ be independent standard Gaussians. The vector
\[
 \sum_{\substack{i \in [m] \\ \rho_i \neq 0}} \frac{\rho_i}{a} Z_i, \quad a Z_1 + \sqrt{1-a^2} W_1, \quad \ldots, \quad a Z_m + \sqrt{1-a^2} W_m
\]
has the same distribution as $\normalg{g}$. Let us denote the first coordinate by $Z_0$, which is some linear combination of $Z_1,\ldots,Z_m$.

We will bound $2s_g^U - 1$. Suppose we are given $q_0,q_1,\ldots,q_m\colon \mathbb{R}^n \to [-1,1]$ such that $\E[q_1] = \cdots = \E[q_m] = 0$. For every $S \subseteq [m]$, we have
\[
 \E\left[q_0(Z_0) \prod_{i \in S} q_i(a Z_i + \sqrt{1-a^2} W_i)\right] =
 \E\left[q_0(Z_0) \prod_{i \in S} (U_a q_i)(Z_i) \right].
\]
Multiplying by $\hat{g}(S)$ and summing over all $S \subseteq [m]$, we obtain
\begin{multline*}
 \E_{\cG \sim \normalg{g}}[q_0(\cG_0) g(q_1(\cG_1),\ldots,q_m(\cG_m))] =
 \E[q_0(Z_0) g((U_a q_1)(Z_1), \ldots, ((U_a q_m)(Z_m)))] \leq \\
 \E[|g((U_a q_1)(Z_1), \ldots, (U_a q_m)(Z_m))|] \leq
 \sqrt{\E[g((U_a q_1)(Z_1), \ldots, (U_a q_m)(Z_m))^2]}.
\end{multline*}
We proceed to bound the expression inside the square root:
\begin{multline*}
 \E[g((U_a q_1)(Z_1), \ldots, (U_a q_m)(Z_m))^2] = \\
 \sum_{S_1,S_2 \subseteq [m]} \hat{g}(S_1) \hat{g}(S_2)
 \E\left[ \prod_{i \in S_1} (U_a q_i)(Z_i) \prod_{j \in S_2} (U_a q_j)(Z_j) \right] =
 \sum_{S \subseteq [m]} \hat{g}(S)^2 \prod_{i \in S} \E[(U_a q_i)^2],
\end{multline*}
since $\E[U_a q_i] = \E[q_i] = 0$ for all $i \in [m]$.
Applying \Cref{thm:borell}, we conclude that
\[
 \E_{\cG \sim \normalg{g}}[q_0(\cG_0) g(q_1(\cG_1),\ldots,q_m(\cG_m))]^2 \leq
 \sum_{S \subseteq [m]} \hat{g}(S)^2 \rho^S, \text{ where } \rho = \frac{2}{\pi} \arcsin(a^2) = \frac{2}{\pi} \arcsin \frac{\|g^{=1}\|^2}{\|g^{\ge 1}\|^2} < 1,
\]
and so
\[
 s_g^U \leq \frac{1}{2} + \frac{1}{2} \sqrt{\langle T_\rho g, g \rangle} < 1.
\]

\begin{remark}
We can generalize the proof technique by choosing $m$ parameters $a_1,\ldots,a_m$ rather than a single parameter $a$. The optimal parameters can be found, in principle, using Lagrange multipliers. When $g$ is symmetric, the optimal choice is $a_1 = \cdots = a_m = a$.
\end{remark}

\subsection{Proof of \crtCref{lem:sg-lb}}

We will assume for simplicity that $s_g^L$ is achieved for some finite $n$; the general case follows using very similar arguments. Also, we will first assume that $q$ is $\{-1,1\}$-valued.

When $\E[g] \neq 0$, define $q_0\colon \mathbb{R}^n \to \{-1,1\}$ by
\[
 q_0(x) = \sgn \left(\E_{\cG \sim \normalg{g}^n}[g(q(\cG_1),\ldots,q(\cG_m)) \mid \cG_0 = x]\right).
\]

We will consider functions on $Nn$ inputs, which we think of as composed of $n$ vectors $x^{(1)},\ldots,x^{(n)} \in \{-1,1\}^N$.
We lift the functions $q,q_0$ from $\mathbb{R}^n$ to $\{-1,1\}^{Nn}$ by defining
\[
 \tilde{q}(x^{(1)},\ldots,x^{(n)}) = q\left(\frac{\sum_{i=1}^N x^{(1)}_i}{\sqrt{N}}, \ldots, \frac{\sum_{i=1}^N x^{(n)}_i}{\sqrt{N}}\right),
\]
and similarly for $q_0$.

If $\E[g] = 0$ then we define $f_{Nn}\colon \{-1,1\}^{Nn} \to \{-1,1\}$ by $f_{Nn} = \tilde{q}$. Otherwise, assume for definiteness that $\E[g] < 0$, and define
\[
 f_{Nn}(x^{(1)},\ldots,x^{(n)}) =
 \begin{cases}
 \tilde{q}(x^{(1)},\ldots,x^{(n)}) & \text{if } \sum_{i=1}^n \sum_{j=1}^N x^{(i)}_j \geq \frac{\E[g]}{2} Nn, \\
 \tilde{q}_0(x^{(1)},\ldots,x^{(n)}) & \text{if } \sum_{i=1}^n \sum_{j=1}^N x^{(i)}_j < \frac{\E[g]}{2} Nn.
 \end{cases}
\]
The central limit theorem shows that $\Pr[f_{Nn} \circ g^{Nn} = g \circ f_{Nn}^m]$ tends to $s_g^L$.

To complete the proof when $q$ is $\{-1,1\}$-valued, we need to show that all low-degree Fourier coefficients of $f_{Nn}$ are small. The total weight of all Fourier coefficients on the ``$(L_1,\ldots,L_n)$ level'' (containing $L_j$ elements from $x^{(j)}$) tends to the corresponding Hermite coefficient of $q$, and so each of these coefficients has magnitude 
\[
 \Theta\left(\frac{1}{\sqrt{\binom{N}{L_1} \cdots \binom{N}{L_n}}}\right) = \Theta\left(\frac{1}{N^{(L_1 + \cdots + L_n)/2}}\right).
\]
In particular, the degree~$L$ Fourier coefficients have magnitude $O(N^{-L/2})$.

Finally, if $q$ is not $\{-1,1\}$-valued, then we first construct fractional ($[-1,1]$-valued) functions $f^*_{Nn}$ with the desired properties. We then convert them to Boolean functions $f_{Nn}$ by independently sampling $f_{Nn}(x)$ for each $x \in \{-1,1\}^{Nn}$ according to the unique distribution supported on $\{-1,1\}$ whose expectation is $f^*_{Nn}$. Standard concentration bounds complete the proof.

\subsection{Proof of \crtCref{lem:sg-nontrivial}}

For $x \in \mathbb{R}$, let
\[
 \gamma(x) = \E_{\cG \sim \normalg{g}}[g(\sgn(\cG_1),\ldots,\sgn(\cG_m)) \mid \cG_0 = x].
\]

In view of \Cref{lem:sg-lb}, it suffices to show that
\[
 \E_{x \sim \stdnormal}[|\gamma(x)|] > |\E[g]|.
\]

As observed in the introduction to \Cref{sec:list-decoding}, the triangle inequality shows that
\[
 \E_{x \sim \stdnormal}[|\gamma(x)|] \geq
 \Bigl| \E_{\cG \sim \normalg{g}}[g(\sgn(\cG_1),\ldots,\sgn(\cG_m))] \Bigr| = |\E[g]|,
\]
with equality only if $|\gamma(x)| = |\E[g]|$ almost surely. Due to continuity, equality is only possible if $\gamma(x) = \E[g]$ for all $x$.

When $x$ is very large in absolute value, it is very likely that $\sgn(\cG_j) = \sgn(\rho_j) \sgn(x)$, where $\rho_j = \E[\cG_0 \cG_j] \neq 0$ by assumption. Hence when $x$ is very large in absolute value, $\gamma(x)$ is close to a $g(y)$ for some $y \in \{-1,1\}^m$, and so $|\gamma(x)| \approx 1$. When $x$ is large enough, $|\gamma(x)| > |\E[g]|$.

\subsection{Proof of \crtCref{thm:sg-AND}}

We will show that the supremum in the definition of $s_g^U$ in \Cref{thm:list-decoding} is attained by $n = 1$ and $q_1 = \cdots = q_m$ being the sign of the first coordinate. \Cref{lem:sg-lb} gives a matching lower bound, showing that in this case, $s_g = s_g^U$.

\smallskip

As shown in the proof of \Cref{thm:list-decoding}, $2s_g^U - 1$ is the maximum of
\[
 \E_{\mathcal{G} \sim \normalg{g}}[q_0(\mathcal{G}_0) g(q_1(\mathcal{G}_1),\ldots,q_m(\mathcal{G}_m))],
\]
and so it suffices to show that this expression is maximized when $q_1,\ldots,q_m$ are sign functions.

We will be using \Cref{thm:neeman}, and to this end we need to switch from $\{-1,1\}$-valued functions to $\{0,1\}$-valued functions. Accordingly, define $p_i(x) = (1+q_i(x))/2$, so that $q_i(x) = 2p_i(x) - 1$. We have
\[
 g(q_1,\ldots,q_m) = 2 \prod_{i=1}^m \frac{1+q_i}{2} - 1 = 2\prod_{i=1}^m p_i - 1,
\]
and so $2s_g^U - 1$ is the maximum of
\[
 \E_{\mathcal{G} \sim \normalg{g}}[(2p_0(\mathcal{G}_0)-1)(2p_1(\mathcal{G}_1)\cdots p_m(\mathcal{G}_m)-1)] =
 4\E_{\mathcal{G} \sim \normalg{g}}[p_0(\mathcal{G}_0) \cdots p_m(\mathcal{G}_m)] - 2\E[p_0] + 1.
\]
Incidentally, the formula for $g$ also shows that $\hat{g}(\{j\}) = 2^{-(m-1)}$, and so $\E[\mathcal{G}_0 \mathcal{G}_j] > 0$.
\Cref{thm:neeman} therefore shows that fixing $\E[p_0]$, the maximum of the expression for $2s_g^U - 1$ is attained when $p_0,\ldots,p_m$ are threshold functions of their first coordinate with the correct expectations. In particular, $p_1(x) = \cdots = p_m(x) = 1_{x_1 \leq 0}$, and so $q_1(x) = \cdots = q_m(x) = \sgn(x_1)$. Maximizing over $\E[p_0]$, we conclude that the supremum in the definition of $s_g^U$ is attained for these $q_1,\ldots,q_m$.

\medskip

When $m = 2$, this shows that
\[
 s_\land = \frac{1}{2} + \frac{1}{2} \E_{x \in \stdnormal} \left[\Bigl| \E_{\cG \sim \normalg{g}}[\sgn(\cG_1) \land \sgn(\cG_2) \mid \cG_0 = x] \Bigr|\right],
\]
where $a \land b = \min(a,b)$.
To calculate this, note that given $\cG_0$, we can sample $\cG_1,\cG_2$ by letting $\cH_0,\cH_1,\cH_2$ be independent standard Gaussians and taking
\[
 \cG_1 = \frac{\cG_0 + \cH_0 + \cH_1}{\sqrt{3}},
 \quad
 \cG_2 = \frac{\cG_0 - \cH_0 + \cH_2}{\sqrt{3}}.
\]
Denoting $\Phi^c(t) = \Pr[\stdnormal > t]$, this shows that for $x \in \mathbb{R}$,
\begin{multline*}
 \E_{\cG \sim \normalg{g}}[\sgn(\cG_1) \land \sgn(\cG_2) \mid \cG_0 = x] =
 2\Pr[\cG_1,\cG_2 > 0] - 1 =
 2\Pr[\cH_0 + \cH_1 > -x, -\cH_0 + \cH_2 > -x] - 1 = \\
 2\E_{y \sim \stdnormal}[\Phi^c(-x-y) \Phi^c(-x+y)] - 1.
\end{multline*}
We conclude that
\[
 s_\land = \frac{1}{2} + \frac{1}{2} \E_{x \sim \stdnormal} \left|
 2\E_{y \sim \stdnormal}[\Phi^c(-x-y) \Phi^c(-x+y)] - 1
 \right| \approx 0.814975356673002.
\]

% sage:
% def Erfc(x): return 1/2*(1 - erf(x/sqrt(2)))
%.5+.5*integral_numerical(lambda x: exp(-x^2/2)/sqrt(2*pi) * abs(-1+2*integral_numerical(lambda y: exp(-y^2/2)/sqrt(2*pi) * Erfc(RR(-x-y)) * Erfc(RR(-x+y)), -oo, oo)[0]), -oo, oo)[0]

\section{Exact classification}
\label{sec:exact}

In this section we classify all exact multi-polymorphisms, that is, all exact solutions to the equation \(f_0 \circ g^n = g \circ (f_1,\ldots,f_m)\), thus proving \Cref{thm:exact-multi-polymorphisms}. As mentioned in the introduction, this was essentially solved by Dokow and Holzman~\cite{DH09}; they gave all solutions to the equation when $f_0,\ldots,f_m$ are \emph{Paretian}, that is, satisfy $f_j(b,\ldots,b) = b$ for $b \in \{0,1\}$. When $g$ is not XOR or NXOR, they gave all solutions under the relaxed assumption that $f_0,\ldots,f_m$ are \emph{surjective}, that is, non-constant.

We provide a new proof that of the characterization. Our proof goes one step further and determines all solutions to the slightly more general equation 
\[
 f_0 \circ g^n = h \circ (f_1,\ldots,f_m),
\]
where $g,h \colon \{-1,1\}^m \to \{-1,1\}$ and $f_0,\ldots,f_m \colon \{-1,1\}^n \to \{-1,1\}$.
(We switched from $\{0,1\}$ to $\{-1,1\}$ in the interest of symmetry.)
In addition, we complete the characterization when the functions $f_0,\ldots,f_m$ are allowed to be constant.

The argument is in two parts. In the first part, we classify (except to some corner cases) all \emph{real} solutions to the equation above. That is, we allow the functions to be real-valued rather than Boolean ($\{-1,1\}$-valued). We extend a function $f\colon \{-1,1\}^n \to \mathbb{R}$ to a function on $\mathbb{R}^n$ multilinearly. That is, we consider the Fourier expansion of $f$,
\[
 f(x) = \sum_{S \subseteq [n]} \hat{f}(S) \prod_{i \in S} x_i,
\]
and extend $f$ to $\mathbb{R}^n$ using this expression.

Equivalently, we think of $f_0,\ldots,f_m,g,h$ as multilinear polynomials, that is, polynomials in which no variable is squared, and interpret the equation above in the following way:
\[
 f_0(g(z_{11},\ldots,z_{1m}), \ldots, g(z_{n1},\ldots,z_{nm})) =
 h(f_1(z_{11},\ldots,z_{n1}), \ldots, f_m(z_{1m},\ldots,z_{nm})),
\]
which is an identity of multilinear polynomials.

In the second part, we determine which of these solutions correspond to Boolean functions.

\subsection{Multilinear classification}
\label{sec:exact-multilinear}

We use the notation $z_{IJ}$ to refer to the monomial $\prod_{i \in I} \prod_{j \in J} z_{ij}$. We also use similar self-explanatory notations such as $z_{iJ},z_{Ij}$. In case $I$ or $J$ are the empty set, $z_{IJ}$ represent the constant monomial~$1$.

Throughout, we assume that $f_0,\ldots,f_m$ are multilinear polynomials over the variables $x_1,\ldots,x_n$; that $g,h$ are multilinear polynomials over the variables $y_1,\ldots,y_m$; and that
\[
 f_0 \circ g^n = h \circ (f_1,\ldots,f_m).
\]

We first identify the set of coordinates that $f_0,\ldots,f_m$ depend on.

\begin{lemma} \label{lem:junta-variables}
Suppose that $f_0,\ldots,f_m$ are non-constant, and that $g,h$ together depend on all coordinates in $[m]$, where $m \geq 1$. Then $g,h$ depend on the same coordinates, and either $\deg(g) = \deg(h) = 1$, or the following holds.

There exists a set $U \subseteq [n]$ such that $f_0,\ldots,f_m$ depend only on the coordinates in $U$, and furthermore
\[
 \hat{f}_0(U),\ldots,\hat{f}_m(U) \neq 0.
\]
If $|U|>1$ then moreover
\[
 \hat{g}([m]),\hat{h}([m]) \neq 0.
\]
\end{lemma}
\begin{proof}
For $j \in [m]$, let $F_i = f_i - \hat{f}_i(\emptyset)$, and let $H(y_1,\ldots,y_m) = h(y_1 + \hat{f}_1(\emptyset), \ldots, y_m + \hat{f}_m(\emptyset))$. Thus $H \circ (F_1,\ldots,F_m) = h \circ (f_1,\ldots,f_m)$. Similarly, let $G = g - \hat{g}(\emptyset)$ and $F_0(x_1,\ldots,x_n) = f_0(x_1 + \hat{g}(\emptyset), \ldots, x_n + \hat{g}(\emptyset))$. Thus $F_0 \circ G^n = f_0 \circ g^n$. Therefore
\begin{equation} \label{eq:L1-FGH}
 F_0 \circ G^n = H \circ (F_1,\ldots,F_m),
\end{equation}
where $\hat{F}_1(\emptyset) = \cdots = \hat{F}_m(\emptyset) = \hat{G}(\emptyset) = 0$; $F_0,\ldots,F_m$ are non-constant; and $G,H$ together depend on all coordinates in $[m]$. We will prove the lemma for the new functions $F_0,\ldots,F_m,G,H$. By inspection, the lemma holds also for the original functions.

Expanding \eqref{eq:L1-FGH} according to the Fourier expansion, and using the condition $\hat{F}_1(\emptyset) = \cdots = \hat{F}_m(\emptyset) = \hat{G}(\emptyset) = 0$, we obtain
\begin{equation} \label{eq:L1-expansion}
 \sum_{S \subseteq [n]} \hat{F}_0(S) \prod_{i \in S} \sum_{\emptyset \neq T_i \subseteq [m]} \hat{G}(T_i) \prod_{j \in T_i} z_{ij} =
 \sum_{T \subseteq [m]} \hat{H}(T) \prod_{j \in T} \sum_{\emptyset \neq S_j \subseteq [n]} \hat{F}_j(S_j) \prod_{i \in S_j} z_{ij}.
\end{equation}
Notice that every monomial appears in exactly one summand on each side.

\paragraph{$G$ depends on all coordinates}
Let $T \neq \emptyset$ be such that $\hat{H}(T) \neq 0$. For every $j \in T$, since $F_j$ is non-constant, $\hat{F}_j(S_j) \neq 0$ for some $S_j \neq \emptyset$. Comparing coefficients of $\prod_{j \in T} z_{S_jj}$ in~\eqref{eq:L1-expansion}, we get
\[
 \hat{F}_0(S) \prod_{i \in S} \hat{G}(T_i) = \hat{H}(T) \prod_{j \in T} \hat{F}_j(S_j),
\]
where $S = \bigcup_{j \in T} S_j$ and $T_i = \{ j \in T : i \in S_j \}$. Thus $\hat{F}_0(S) \neq 0$ and $\hat{G}(T_i) \neq 0$ for all $i \in S$. For every $j \in T$, since $S_j$ is non-empty, we can find $i \in S_j$, and so $j \in T_i$. Consequently, if $H$ depends on some $j \in [m]$ then $G$ also depends on $j$. This implies that $G$ depends on all coordinates.

\paragraph{$H$ depends on all coordinates}
Since $F_0$ is non-constant, $\hat{F}_0(S) \neq 0$ for some $S \neq \emptyset$. Let $T \neq \emptyset$. Comparing coefficients of $z_{ST}$ in~\eqref{eq:L1-expansion}, we get
\[
 \hat{F}_0(S) \hat{G}(T)^{|S|} = \hat{H}(T) \prod_{j \in T} \hat{F}_j(S).
\]
Since $G$ depends on all coordinates, for every $j \in [m]$ there exists a set $T$ containing $j$ such that $\hat{G}(T) \neq 0$. This shows that $\hat{F}_j(S) \neq 0$ for all $j \in [m]$.
We conclude that if $T \neq \emptyset$ then
$\hat{G}(T) \neq 0$ iff $\hat{H}(T) \neq 0$. In particular, $G,H$ both depend on all coordinates, and $\deg(G) = \deg(H) > 1$.

This argument also shows that if $\hat{F}_0(S) \neq 0$ then $\hat{F}_j(S) \neq 0$ for all $j \in [m]$.

\paragraph{There exists a unique inclusion-maximal set $U$ in the support of $\hat{F}_0$}
Let $U$ be an inclusion-maximal set in the support of the Fourier expansion of $F_0$. Thus $\hat{F}_0(U),\hat{F}_1(U),\ldots,\hat{F}_m(U) \neq 0$. We claim that $F_0$ depends only on the coordinates in $U$. Otherwise, $\hat{F}_0(V) \neq 0$ for some $V \nsubseteq U$. Since $\deg(H)>1$, there exists a set $T$ of size at least~$2$ such that $\hat{H}(T) \neq 0$. Choose an arbitrary $j \in T$. The coefficient of $z_{Uj} z_{V(T \setminus \{j\})}$ on the right-hand side of~\eqref{eq:L1-expansion} is
\[
 \hat{H}(T) \hat{F}_j(U) \prod_{k \in T \setminus \{j\}} \hat{F}_k(V) \neq 0.
\]
The coefficient of the same monomial on the left-hand side of~\eqref{eq:L1-expansion} is a multiple of $\hat{F}_0(U \cup V)$, showing that $\hat{F}_0(U \cup V) \neq 0$, which contradicts the choice of $U$.

\paragraph{$F_j$ depends only on $U$}
Let $j \in [m]$. We claim that $F_j$ also depends only on the coordinates in $U$. Otherwise, $\hat{F}_j(V) \neq 0$ for some $V \nsubseteq U$. Since $H$ depends on all coordinates, there exists a set $T$ containing $j$ such that $\hat{H}(T) \neq 0$. Comparing coefficients of $z_{Vj} z_{U(T \setminus \{j\})}$, as before we obtain that either $\hat{F}_0(V) \neq 0$ (if $|T| = 1$) or $\hat{F}_0(V \cup U) \neq 0$ (if $|T| > 1$), which contradicts the fact that $F_0$ depends only on the coordinates in $U$.

\paragraph{If $|U| > 1$ then $\hat{G}([m]),\hat{H}([m]) \neq 0$}
We now show that if $|U| > 1$ then the support of the Fourier expansion of $G$ is closed under union, and so $\hat{G}([m]) \neq 0$ since $G$ depends on all coordinates. This implies that $\hat{H}([m]) \neq 0$ as well.

Let $i \in U$, and let $T_1,T_2 \neq \emptyset$ be two different sets in the support of the Fourier expansion of $G$. The coefficient of $z_{iT_1} z_{(U \setminus [i])T_2}$ on the left-hand side of~\eqref{eq:L1-expansion} is
\[
 \hat{F}_0(U) \hat{G}(T_1) \hat{G}(T_2)^{|U|-1} \neq 0.
\]
The coefficient of the same monomial on the right-hand side of~\eqref{eq:L1-expansion} is a multiple of $\hat{H}(T_1 \cup T_2)$, showing that $\hat{H}(T_1 \cup T_2) \neq 0$ and so $\hat{G}(T_1 \cup T_2) \neq 0$. The support of the Fourier expansion of $G$ is thus closed under union. \end{proof}

We now determine a formula for all functions, except for some corner cases.

\begin{lemma} \label{lem:formula}
Suppose that $f_0,\ldots,f_m$ are non-constant, that $g,h$ together depend on all coordinates in $[m]$, where $m \geq 1$, and that $\deg(g),\deg(h) > 1$. Let $U$ be the set promised by \Cref{lem:junta-variables}, and suppose that $|U| > 1$.

There exist constants $A_0,\ldots,A_m,C,D \neq 0$ and $\kappa_0,\ldots,\kappa_m,B_0,\ldots,B_m$ such that
\begin{align*}
g(y) &= C \prod_{j \in [m]} (y_j + \kappa_j) - \kappa_0,
&	
h(y) &= D \prod_{j \in [m]} (y_j + B_j) - B_0,
\\
f_0(x) &= A_0 \prod_{i \in U} (x_i + \kappa_0) - B_0,
&
f_j(x) &= A_j \prod_{i \in U} (x_i + \kappa_j) - B_j,
\end{align*}
where
\[
 A_0 C^{|U|} = D \prod_{j \in [m]} A_j.
\]
\end{lemma}
\begin{proof}
Suppose for definiteness that $U = [n]$, where $n \geq 2$.

\paragraph{Reducing $g,f_1,\ldots,f_m$ to affine shifts of monomials}
Let $\kappa_1,\ldots,\kappa_m$ be parameters to be chosen later. For $j \in [m]$, define $F_j(x_1,\ldots,x_n) = f_j(x_1 - \kappa_j,\ldots,x_n - \kappa_j)$, and let $G(y_1,\ldots,y_m) = g(y_1 - \kappa_1, \ldots, y_m - \kappa_m)$. Notice that
\[
 f_0 \circ G^n = h \circ (F_1,\ldots,F_m),
\]
which implies that
\begin{equation} \label{eq:L2-expansion}
 \sum_{S \subseteq [n]} \hat{f}_0(S) \prod_{i \in S} \sum_{T_i \subseteq [m]} \hat{G}(T_i) \prod_{j \in T_i} z_{ij} =
 \sum_{T \subseteq [m]} \hat{h}(T) \prod_{j \in T} \sum_{S_j \subseteq [n]} \hat{F}_j(S_j) \prod_{i \in S_j} z_{ij}.
\end{equation}

We choose $\kappa_j$ in such a way that $\hat{F}_j([n-1]) = 0$. This is possible since
\[
 \hat{F}_j([n-1]) = \hat{f}_j([n-1]) - \kappa_j \hat{f}_j([n]),
\]
and by \Cref{lem:junta-variables}, $\hat{f}_j([n]) \neq 0$. The lemma also implies that $\hat{G}([m]),\hat{h}([m]) \neq 0$.

Let $T \neq \emptyset,[m]$. Consider the monomial $\zeta = z_{[n-1][m]} z_{nT} = z_{[n-1] \overline{T}} z_{[n] T}$. Since $\zeta$ mentions all row indices and all column indices, it only appears once on each side of~\eqref{eq:L2-expansion} (for $S = [n]$ and $T = [m]$). Comparing coefficients, we obtain
\[
 \hat{f}_0([n]) \hat{G}([m])^{n-1} \hat{G}(T) = \hat{h}([m]) \prod_{j \in T} \hat{F}_j([n]) \prod_{j \notin T} \hat{F}_j([n-1]) = 0.
\]
It follows that $\hat{G}(T) = 0$ for all $T \neq \emptyset,[m]$. Therefore there exist constants $C \neq 0$ and $\kappa_0$ such that
\[
 G(y_1,\ldots,y_m) = C y_{[m]} - \kappa_0.
\]
Every monomial appearing on the left-hand side of~\eqref{eq:L2-expansion} is thus of the form $z_{S[m]}$ for some $S \subseteq [n]$.

Let $S_1,\ldots,S_m \neq \emptyset$ be such that $\hat{F}_j(S_j) \neq 0$. The coefficient of $\prod_{j \in [m]} z_{S_jj}$ on the right-hand side of~\eqref{eq:L2-expansion} is
\[
 \hat{h}([m]) \prod_{j \in [m]} \hat{F}_j(S_j) \neq 0,
\]
and so necessarily $S_1 = \cdots = S_m$. Since we know that $\hat{F}_j([n]) \neq 0$ for all $j \in [m]$, it follows that the Fourier expansions of $F_1,\ldots,F_m$ are supported on $\emptyset,[n]$, and so there exist $A_1,\ldots,A_m \neq 0$ and $B_1,\ldots,B_m$ such that
\[
 F_j(x_1,\ldots,x_n) = A_j x_{[n]} - B_j.
\]

\paragraph{Reducing $f_0,h$ to affine shifts of monomials}
Let $F_0(x_1,\ldots,x_n) = f_0(Cx_1-\kappa_0,\ldots,Cx_n-\kappa_0)$, so that
\[
 F_0 \circ y_{[m]}^n = f_0 \circ G^n.
\]
Similarly, let $H(y_1,\ldots,y_m) = h(A_1y_1-B_1,\ldots,A_my_m-B_m)$, so that
\[
 H \circ (x_{[n]},\ldots,x_{[n]}) = h \circ (F_1,\ldots,F_m).
\]
Thus
\[
 F_0(z_{1[m]},\ldots,z_{n[m]}) = H(z_{[n]1},\ldots,z_{[n]m}).
\]
Every monomial appearing on the left-hand side is of the form $z_{S[m]}$, and every monomial appearing on the right-hand side is of the form $z_{[n]T}$. Consequently, the only monomials appearing in this equation are $1$ and $z_{[n][m]}$. Therefore there exist $\alpha,\beta \neq 0$ and $\gamma,\delta$ such that
\begin{align*}
 F_0(x) &= \alpha x_{[n]} - \gamma, &
 H(y) &= \beta y_{[m]} - \delta.
\end{align*}
Inspection shows that $\alpha = \beta$ (which correspond to $A_0,D$) and $\gamma = \delta$ (which correspond to $B_0$).

Unrolling, we get the formulas in the statement of the lemma.
\end{proof}

\subsection{Boolean classification}
\label{sec:exact-boolean}

Since we are mainly interested in Boolean functions, we need to determine when functions of the form given by \Cref{lem:formula} are Boolean.

\begin{lemma} \label{lem:boolean}
Suppose that $f\colon \{-1,1\}^n \to \{-1,1\}$ is given by
\[
 f = A \prod_{i=1}^n (x_i + \kappa_i) - B,
\]
where $A \neq 0$ and $n \geq 1$. Then one of the following holds:
\begin{itemize}
\item $\kappa_1 = \cdots = \kappa_n = B = 0$ and $A \in \{-1,1\}$, i.e.
\[
 f = A \prod_{i=1}^n x_i, \text{ where } A \in \{-1,1\}.
\]
\item $\kappa_1,\ldots,\kappa_n,B \in \{-1,1\}$ and $A = 2B/\prod_{i=1}^n (2\kappa_i)$, i.e.
\[
 f = 2B \prod_{i=1}^n \frac{\kappa_i x_i + 1}{2} - B, \text{ where } \kappa_1,\ldots,\kappa_n,B \in \{-1,1\}.
\]
\end{itemize}
\end{lemma}
\begin{proof}
If $n = 1$ then $f \in \{x_1,-x_1\}$. Now suppose that $n \geq 2$.
Choose an assignment $y_i \neq -\kappa_i$ for all $i \in [n]$, and let $f(y) = F$, where $F + B \neq 0$. Let $\rho_i = (\kappa_i - y_i)/(\kappa_i + y_i)$, and let $y_S$ be the point different from $y$ on the indices in $S \subseteq [n]$. Then
\[
 f(y_S) = (F + B) \prod_{i \in S} \rho_i - B.
\]
Since $f(y_S) \in \{-1,1\}$ for all $S$, the products $\prod_{i \in S} \rho_i$ must attain exactly two different values. In particular, considering singleton $S$, since $\rho_i \neq 1$ for all $i$, we deduce that all $\rho_i$ are equal to some $\rho \neq 1$. If $\rho \neq 0$ then $\rho^2 \neq \rho$, hence $\rho^2 = 1$, implying that $\rho = -1$. Furthermore, $-F = f(y_{\{1\}}) = (F+B)\rho - B$, and so
\[
 \rho = \frac{B-F}{B+F}, \quad \kappa_i = \frac{1+\rho}{1-\rho} y_i = \frac{B}{F} y_i.
\]

If $\rho = -1$ then $B = 0$ and so $\kappa_1 = \cdots = \kappa_m = 0$.
If $\rho = 0$ then $B = F$ and so $\kappa_i = y_i$. In both cases, we can compute $A$ via $f(y) = F$.
\end{proof}

Putting everything together, we obtain our main classification.

\begin{theorem} \label{thm:main-g-h}
Let $f_0,\ldots,f_m\colon \{-1,1\}^n \to \{-1,1\}$ and $g,h\colon \{-1,1\}^m \to \{-1,1\}$ satisfy
\[
 f_0 \circ g^n = h \circ (f_1,\ldots,f_m).
\]

Let $J \subseteq [m]$ consist of all $j$ such that $f_j$ is non-constant, and let $H\colon \{-1,1\}^J \to \{-1,1\}$ be obtained from $h$ by fixing the $\overline{J}$-coordinates to the constants $f_j$ for $j \in \overline{J}$.

One of the following holds:
\begin{enumerate}[(i)]
\item $f_0$ is constant, and $H = f_0$ is also constant. \label{c:f0-H-const}
\item $g$ and $H$ are constants, and $f_0(g,\ldots,g) = H$. \label{c:g-H-const}
\item There exists $j \in J$ such that $g(y) = \gamma y_j$ and $H(y) = \eta y_j$, and furthermore $f_0(\gamma x) = \eta f_j(x)$. \label{c:g-H-dict}
\item There exists $i \in [n]$ such that $f_j = \phi_j x_i$ for $j \in \{0,\ldots,m\}$, and furthermore $\phi_0 g(y_1,\ldots,y_m) = H(\phi_1 y_1, \ldots, \phi_m y_m)$. \label{c:f-dict}
\item There exist non-empty sets $K \subseteq J$ and $U \subseteq [n]$ such that one of the following holds: \label{c:prod}
\begin{enumerate}[(a)]
\item $g(y) = \gamma y_K$, $H(y) = \eta y_K$ and $f_j(x) = \phi_j x_U$ for $j \in K \cup \{0\}$, where $\phi_0 \gamma^{|U|} = \eta \prod_{j \in K} \phi_j$. \label{c:XOR}
\item There exist $\kappa_j,B_j \in \{-1,1\}$ for $j \in K \cup \{0\}$ such that \label{c:AND}
\begin{align*}
g(y) &= 2\kappa_0 \prod_{j \in K} \frac{1 + \kappa_j y_j}{2} - \kappa_0,
&	
H(y) &= 2B_0 \prod_{j \in K} \frac{1 + B_j y_j}{2} - B_0,
\\
f_0(x) &= 2B_0 \prod_{i \in U} \frac{1 + \kappa_0 x_i}{2} - B_0,
&
f_j(x) &= 2B_j \prod_{i \in U} \frac{1 + \kappa_j x_i}{2} - B_j.
\end{align*}
\end{enumerate}
\end{enumerate}
\end{theorem}

\begin{proof}
If $f_0$ is constant and $y \in \{-1,1\}^J$, then since the functions $f_j$ are not constant for $j \in J$, we can find an input $z$ such that $f_0 = (h \circ (f_1,\ldots,f_m))(z) = H(y)$, showing that $H$ is also constant. This is case~\eqref{c:f0-H-const}.

Now suppose that $f_0$ is not constant. We claim that $g$ depends only on the coordinates in $J$. Otherwise, we can find two inputs $z^{(0)},z^{(1)}$, differing only on the coordinates in $\overline{J}$, such that $g(z_0) = 0$ and $g(z_1) = 1$. Thus $f_0(g(z^{(y_1})),\ldots,g(z^{(y_n)})) = f_0(y_1,\ldots,y_n)$. However, $f_0 \circ g^n = h \circ (f_1,\ldots,f_m)$ is invariant to changing only coordinates in $\overline{J}$, showing that $f_0$ is constant, contradicting our assumption.
From now on we think of $g$ as a function on $\{-1,1\}^J$.

Let $K \subseteq J$ be the set of coordinates that $g,H$ together depend on.
According to \Cref{lem:junta-variables}, both of $g,H$ depends on all coordinates in $K$.
If $K = \emptyset$ then $g,H$ are both constants, which is case~\eqref{c:g-H-const}.
If $K = \{j\}$ then $g,H$ are both of the form $\pm y_j$, which is case~\eqref{c:g-H-dict}.

If $|K| > 1$ then $\deg(g),\deg(H) > 1$, since a Boolean function of degree~$1$ depends on a single coordinate. Therefore \Cref{lem:junta-variables} shows that there exists a non-empty set $U \subseteq [n]$ such that $f_0,\ldots,f_m$ depends only on the coordinates in $U$, and $\hat{f}_0(U),\ldots,\hat{f}_m(U) \neq 0$. If $|U|=1$, say $U = \{i\}$, then each of $f_0,\ldots,f_m$ is of the form $\pm x_i$, which is case~\eqref{c:f-dict}.
If $|U| > 1$ then \Cref{lem:formula} shows that we are in case~\eqref{c:prod}:
there exist constants $A_j,C,D \neq 0$ and $\kappa_j,B_j$ such that
\begin{align*}
g(y) &= C \prod_{j \in K} (y_j + \kappa_j) - \kappa_0,
&	
H(y) &= D \prod_{j \in K} (y_j + B_j) - B_0,
\\
f_0(x) &= A_0 \prod_{i \in U} (x_i + \kappa_0) - B_0,
&
f_j(x) &= A_j \prod_{i \in U} (x_i + \kappa_j) - B_j.
\end{align*}

\Cref{lem:boolean} applied to $f_0,\ldots,f_m$ shows that for each $j \in K \cup \{0\}$, either $\kappa_j = B_j = 0$ or $\kappa_j,B_j \in \{-1,1\}$. Applying the same lemma to $g$ and $H$, we see that in fact either all $\kappa_j$ and all $B_j$ equal~$0$, or all of them belong to $\{-1,1\}$. The former is subcase~\eqref{c:XOR}, and the latter is subcase~\eqref{c:AND}.
\end{proof}

Let us now switch back from $\{-1,1\}$ to $\{0,1\}$. We do so using the correspondence
\[
 0 \leftrightarrow 1, \quad 1 \leftrightarrow -1.
\]

Under this correspondence, product corresponds to XOR; the function
\[
 2 \prod_{i \in U} \frac{1 + (-1)^{b_i} x_i}{2} - 1
\]
corresponds to $\biglor_{i \in U} (x_i \oplus b_i)$; and the function
\[
 -2 \prod_{i \in U} \frac{1 - (-1)^{b_i} x_i}{2} + 1
\]
corresponds to $\bigland_{i \in U} (x_i \oplus b_i)$. 

Using this, we can translate \Cref{thm:main-g-h} into $\{0,1\}$ variables, and deduce \Cref{thm:exact-multi-polymorphisms} by taking $g = h$.

\section{Open questions}
\label{sec:open-questions}

Our work suggests many open questions. Here are some of them.

\begin{open}
Can \Cref{thm:main} be extended to polymorphisms of predicates? That is, given a function $g\colon \{0,1\}^m \to \{0,1\}$, what we can say about functions $f\colon \{0,1\}^n \to \{0,1\}$ satisfying
\[
 \Pr[(g \circ f^m)(Z) = 1 \mid g(\row_i(Z)) = 1 \text{ for all } i \in [n]] \geq 1 - \epsilon?
\]
\end{open}

As mentioned in the introduction, Kalai~\cite{Kalai} proved a version of \Cref{thm:main} for the predicate $\mathsf{NAE_3}$, and Friedgut and Regev proved a version of \Cref{thm:main} for the predicate $\mathsf{NAND_2}$.

\begin{open}
What is the optimal dependence between $\epsilon$ and $\delta$ in \Cref{thm:main}? Does it depend on $g$?
\end{open}

In our current proof, the dependence is not polynomial. In fact, due to the use of Jones' regularity lemma, the dependence is of tower type. This can be dramatically improved by using a different regularity lemma, which approximates the function by a decision tree rather than by a junta. The dependence now becomes only doubly exponential. We sketch this argument in \Cref{apx:decision-trees}.

For many specific $g$ we can prove a version of \Cref{thm:main} in which $\delta$ is polynomial in $\epsilon$. This is the case for linearity testing, and also for $\Maj_3$, the majority function on three inputs, as we sketch in \Cref{apx:Maj3}.

\begin{open}
Can we extend \Cref{thm:main} to larger alphabets, replacing $\{0,1\}$ with an arbitrary finite set?
\end{open}

One issue is that, to the best of our knowledge, a complete analog of \Cref{thm:exact-multi-polymorphisms} for larger alphabets is not currently known, though some preliminary results appear in \cite{DH10,SX}. Moreover, while the complete classification of polymorphism of binary predicates is known (it is given by Post's lattice), the situation for larger alphabets is known to be much wilder.

Nevertheless, it might be possible to show that every approximate polymorphism is close to a skew polymorphism, even without classifying the latter.

\begin{open}
Can we extend \Cref{thm:main} to tensors? For example, what can we say about Boolean functions $f,g,h$ satisfying $f \circ (g \circ h^m)^n = g \circ (h \circ f^p)^m$ with probability $1 - \delta$?
\end{open}

\Cref{sec:list-decoding} gives an upper bound $s_g^U$ and a lower bound $s_g^L$ on $s_g$ which are similar but not identical.

\begin{open}
Is $s_g^U = s_g^L$? Is the optimum always achieved in one-dimensional Gaussian space?
\end{open}

Another interesting question concerns an analog of ``approximation resistance''.
When $g$ is unbalanced, we trivially have $s_g \ge \max(\E[g],1-\E[g])$ by taking $f$ to be random around the middle slice, and constant around the $\E[g]$-slice; and when $g$ is balanced, we trivially have $s_g \ge 1/2$ by taking $f$ to be a random function.

\begin{open}
For which functions $g$ is $s_g > \max(\E[g],1-\E[g])$?
\end{open}

\Cref{lem:sg-nontrivial} shows that the strict inequality holds for unbalanced $g$ whenever all Fourier coefficients on the first level are non-zero. Conversely, when all Fourier coefficients on the first level vanish, \Cref{thm:list-decoding} shows that equality holds.

Finally, it would be nice to extend the classification of exact solutions to the case in which we are allowed not only multiple $f$'s, but also multiple $g$'s.

\begin{open}
Classify all solutions $f_0,\ldots,f_m\colon \{0,1\}^n \to \{0,1\}$ and $g_0,g_1,\ldots,g_n\colon \{0,1\}^m \to \{0,1\}$ to the equation
\[
 f_0 \circ (g_1,\ldots,g_n) = g_0 \circ (f_1,\ldots,f_m).
\]
\end{open}

We conjecture that except for some corner cases, the solutions are either XORs or ANDs/ORs of literals.

\bibliographystyle{alphaurl}
\bibliography{biblio}

\appendix

\section{Regularity lemma for multiple measures}
\label{apx:jones}

In this section we prove \Cref{thm:jones} and \Cref{thm:jones-dt}, closely following the argument of Jones~\cite{Jones}. It will be more convenient to consider functions taking values in $\{-1,1\}$ rather than $\{0,1\}$.

Jones' argument employs the notion of \emph{noise stability}, closely related to noise sensitivity.

\begin{definition}[Noise stability] \label{def:noise-stability}
Let $p,\rho \in (0,1)$ and let $f\colon \{0,1\}^n \to \{-1,1\}$. The \emph{noise stability} of $f$ with respect to $\mu_p$ is
\[
 \Stab_\rho^{\mu_p}(f) = \E_{x,y \sim N_\rho^{\mu_p}}[f(x)f(y)] = \sum_{S \subseteq [n]} \rho^{|S|} \hat{f}(S)^2.
\]
\end{definition}

\Cref{thm:jones} is formulated for low-degree influences. We will work directly with a different notion.

\begin{definition}[Noisy influence]
\label{def:noisy-influence}

Let $p,\rho \in (0,1)$, and let $f\colon \{0,1\}^n \to \mathbb{R}$. For $i \in [n]$, define $\partial_i f\colon \{0,1\}^{n-1} \to \mathbb{R}$ by $\partial_i f = f_{\{i\} \to 1} - f_{\{i\} \to 0}$.

The $i$th \emph{noisy influence} of $f$ with respect to $\mu_p$ is
\[
 \Inf_i^{\rho,\mu_p}[f] = p(1-p) \Stab_\rho^{\mu_p}[\partial_i f] =
 \sum_{\substack{S \subseteq [n] \\ i \in S}} \rho^{|S|-1} \hat{f}(S)^2.
\]

A function $f\colon \{0,1\}^n \to \{-1,1\}$ is \emph{$(\delta,\tau)$-noisy-regular with respect to $\mu_p$} if $\Inf_i^{(1-\delta),\mu_p}[f] \leq \tau$ for all $i \in [n]$.
\end{definition}

The following lemma drives the entire proof. It shows that if some noisy influence is large, then querying the corresponding variable increases the average noise stability.

\begin{lemma} \label{lem:jones}
Let $p,\rho \in (0,1)$, and let $f\colon \{0,1\}^n \to \{-1,1\}$. For every $i \in [n]$,
\[
 \E_{a \sim \mu_p}[\Stab_\rho^{\mu_p}(f_{\{i\} \to a})] = \Stab_\rho^{\mu_p}(f) + (1-\rho) \Inf_i^{\rho,\mu_p}[f] \geq \Stab_\rho^{\mu_p}(f).
\]
\end{lemma}
\begin{proof}
Assume for simplicity that $i = n$.

The left-hand side is the expectation of $f(x) f(y)$ with respect to $(x,y) \sim N_\rho^{\mu_p}$, conditioned on $x_n = y_n$. Equivalently, it is the expectation of $f(x) f(y^{n \to x_n})$, where $y^{n \to x_n}$ is obtained by setting the $n$th coordinate to $x_n$. Therefore the left-hand side equals
\[
 \E_{(x,y) \sim N_\rho^{\mu_p}}[f(x) f(y) + f(x) (f(y^{n \to x_n}) - f(y))] =
 \Stab_\rho^{\mu_p}(f) + \E_{(x,y) \sim N_\rho^{\mu_p}}[f(x) (f(y^{n \to x_n}) - f(y))].
\]
The second summand equals
\begin{multline*}
 (1-p)\cdot(1-\rho)\cdot p \E_{(x',y') \sim N_\rho^{\mu_p}}[f(x',0)(f(y',0) - f(y',1))] +
 p\cdot(1-\rho)\cdot (1-p) \E_{(x',y') \sim N_\rho^{\mu_p}}[f(x',1)(f(y',1) - f(y',0))] = \\
 (1-\rho)p(1-p) \E_{(x',y') \sim N_\rho^{\mu_p}}[(f(x',0) - f(x',1))(f(y',0) - f(y',1))] = (1-\rho) \Inf_i^{\rho,\mu_p}[f],
\end{multline*}
completing the proof.
\end{proof}

Using this lemma, we prove a version of \Cref{thm:jones} for noisy influences.

\begin{lemma} \label{lem:jones-noisy-influences}
For all $\delta,\tau,\epsilon > 0$ and $p_1,\ldots,p_\ell \in (0,1)$ there exists $L \in \mathbb{N}$ such that the following holds.

For every $f\colon \{0,1\}^n \to \{-1,1\}$ we can find a set $T \subseteq [n]$ of at most $L$ coordinates such that for all $k \in [\ell]$,
\[
 \Pr_{z \in \mu_{p_k}(\{0,1\}^T)}[f_{T \to z} \text{ is $(\delta,\tau)$-noisy-regular with respect to $\mu_{p_k}$}] \geq 1-\epsilon.
\]
\end{lemma}
\begin{proof}
We define a potential function on subsets of $[n]$:
\[
\phi(T) = \sum_{k=1}^\ell \E_{z \sim \mu_{p_k}(\{0,1\}^T)}[\Stab_{1-\delta}^{\mu_{p_k}}(f_{T \to z})].
\]
Clearly $\phi(\emptyset) \ge 0$ and $\phi(T) \leq \ell$ for all $T \subseteq [n]$.

We define an increasing sequence of subsets, starting with $T_0 = \emptyset$, and terminating with a set which satisfies the requirements for $T$ in the lemma. 

Suppose we have constructed $T_N$. If $T_N$ satisfies the lemma, then we are done. Otherwise, there exists some $k \in [\ell]$ such that with respect to $\mu_{p_k}$, the function $f_{T_N \to z}$ is not $(\delta,\tau)$-noisy-regular with probability at least $1-\epsilon$.
For each $z \in \{0,1\}^{T_N}$ such that $f_{T_N \to z}$ is not $(\delta,\tau)$-noisy-regular with respect to $\mu_{p_k}$, choose a variable $i_z$ satisfying $\Inf_{i_z}^{(1-\delta),\mu_{p_k}}[f_{T_N \to z}] > \tau$, and let $T_{N+1}$ consist of $T_N$ together with all of the variables $i_z$. Thus $|T_{N+1}| \leq |T_N| + 2^{|T_N|}$.

\Cref{lem:jones} shows that
\[
 \phi(T_{N+1}) \geq \phi(T_N) + \epsilon \delta \tau.
\]
It follows that the process terminates at some $N \leq \ell/\epsilon\delta\tau$, and we can choose $L$ accordingly.
\end{proof}

We now derive \Cref{thm:jones} by relating noisy influences and low degree influences.
\begin{proof}[Proof of \Cref{thm:jones}]
Given parameters $d,\tau,\epsilon,p_1,\ldots,p_\ell$, let $\delta = 1/d$, and apply \Cref{lem:jones-noisy-influences} with $(1-\delta)^d \tau = \Theta(\tau)$ to obtain a set $T$.

Choose an arbitrary $p \in \{p_1,\ldots,p_\ell\}$. According to \Cref{lem:jones-noisy-influences}, if $z \sim \mu_p$ then $f_{T \to z}$ is $(\delta,\tau)$-noise-regular with respect to $\mu_p$ with probability $1-\epsilon$. When this happens, for every $i \in [n] \setminus T$ we have
\[
 \Inf_i^{\leq d}[f_{T \to z}] = \sum_{\substack{|S| \leq d \\ i \in S}} \widehat{f_{T \to z}}(S)^2 \leq (1-\delta)^{-d} \sum_{\substack{S \subseteq [n] \\ i \in S}} (1-\delta)^{|S|-1} \widehat{f_{T \to z}}(S)^2 \leq \tau,
\]
and so $f_{T \to z}$ is also $(d,\tau)$-regular.
\end{proof}

%%%

\medskip

For completeness, we now formulate and prove a version of \Cref{thm:jones} for decision trees.

\begin{definition}[Random leaf] \label{def:random-leaf}
Let $T$ be a decision tree on $\{0,1\}^n$ in which internal nodes are labelled with variables.

Given $p \in (0,1)$, a \emph{random leaf} is a pair $(B,z)$, where $B \subseteq [n]$ and $z \in \{0,1\}^B$, formed by choosing a random path in $T$ according to $\mu_p$. We denote this distribution by $\mu_p(T)$.
\end{definition}

\begin{theorem}
\label{thm:jones-dt}

Fix $p_1,\ldots,p_\ell \in (0,1)$. For all $d \in \mathbb{N}$ and $\tau,\epsilon > 0$ there exists $L \in \mathbb{N}$ such that the following holds.

For every $f\colon \{0,1\}^n \to \{-1,1\}$ we can find a decision tree $T$ of depth at most $L$ such that for all $k \in [\ell]$
\[
 \Pr_{(B,z) \sim \mu_{p_k}(T) \in \{0,1\}^T}[f_{B \to z} \text{ is $(d,\tau)$-regular with respect to $\mu_{p_k}$}] \geq 1-\epsilon.
\]

Furthermore, we can take $L = O(\ell d/\epsilon\tau)$.
\end{theorem}

The proof of \Cref{thm:jones-dt} is very similar to the proof of \Cref{thm:jones}, so we only sketch it. We start with an analog of \Cref{lem:jones-noisy-influences}.
\begin{lemma} \label{lem:jones-noisy-influences-dt}
For all $\delta,\tau,\epsilon > 0$ and $p_1,\ldots,p_\ell \in (0,1)$ there exists $L \in \mathbb{N}$ such that the following holds.

For every $f\colon \{0,1\}^n \to \{-1,1\}$ we can find a decision tree $T$ of depth at most $L$ such that for all $k \in [\ell]$,
\[
 \Pr_{(B,z) \in \mu_{p_k}(T)}[f_{B \to z} \text{ is $(\delta,\tau)$-noisy-regular with respect to $\mu_{p_k}$}] \geq 1-\epsilon.
\]

Furthermore, we can take $L = \ell/\epsilon\delta\tau$.
\end{lemma}
\begin{proof}
The proof is almost identical to the proof of \Cref{lem:jones-noisy-influences}. We define the potential function in an analogous way, for decision trees:
\[
 \phi(T) = \sum_{k=1}^\ell \E_{(B,z) \sim \mu_{p_k}(T)}[\Stab_{1-\delta}^{\mu_{p_k}}(f_{B\to z})].
\]

We construct a sequence of decision trees, starting with a decision tree $T_0$ consisting of only a root. If $T_N$ doesn't satisfy the requirements in the lemma, then we can extend it to a decision tree $T_{N+1}$ which is one level deeper and satisfies $\phi(T_{N+1}) \geq \phi(T_N) + \epsilon\delta\tau$.

The process must stop within $\ell/\epsilon\delta\tau$ steps, and the final tree has at most this depth.
\end{proof}

\Cref{thm:jones-dt} now immediately follows, using an identical argument to the proof of \Cref{thm:jones}; since we choose $\delta = 1/d$ and apply \Cref{lem:jones-noisy-influences-dt} with $\Theta(\tau)$, the resulting decision tree has depth $O(\ell d/\epsilon\tau).$

\section{Noise resilience of connected distributions}
\label{apx:connected}

In this section, we prove \Cref{thm:connected}, using ideas from Mossel~\cite{Mossel2010} (especially Lemma~6.1).
Mossel in fact proves a version of \Cref{thm:connected} in his Lemma~6.2, but his result needs a different set of bipartite graphs to be connected, namely the ones in which $X$ is the projection to a single coordinate, and $Y$ is the projection to all other coordinates. Our distributions $\normalg{g}$ are not connected in this sense since if we choose $i$ to be the coordinate corresponding to $g(x)$, then the resulting bipartite graph is a disconnected matching $\{ (g(x),x) : x \in \{-1,1\}^m \}$. 

We will need several preliminaries.

\begin{lemma}[Hypercontractivity] \label{lem:hypercontractivity}
Let $\mathcal{D}$ be a connected distribution on $\{-1,1\}^k$, with non-constant marginals $\mathcal{D}_1,\ldots,\mathcal{D}_k$. Let $\lambda_i = \min(\Pr[\mathcal{D}_i = 1], \Pr[\mathcal{D}_i = -1])$. The following holds for $M = \max(1/\lambda_1,\ldots,1/\lambda_k)$.

If $f\colon \{-1,1\}^n \to \mathbb{R}$ then for every $i \in [k]$ and every $q \geq 2$,
\[
 \E[|f(\mathcal{D}_i^n)|^q]^{1/q} \leq \left(M^{1/2-1/q} \sqrt{q-1}\right)^{\deg(f)} \E[|f(\mathcal{D}_i^n)|^2]^{1/2}.
\]
\end{lemma}
\begin{proof}
This is \cite[Corollary 10.20]{ODonnell} followed by a standard argument.
\end{proof}

The Efron--Stein decomposition is an analog of the Fourier decomposition for arbitrary distributions $\mathcal{D}$ on arbitrary domains $X$ (in our case, $X \subseteq \{-1,1\}^k$). It decomposes an arbitrary function $f$ on $X^n$ as a sum
\[
 f = \sum_{S \subseteq [n]} f_S,
\]
where $f_S$ only depends on the coordinates in $S$; if we fix coordinates in any set $R \not\supseteq S$ then $f_S$ averages to zero; and $f_S,f_T$ are orthogonal for $S \neq T$ (with respect to $\mu$). We define $\deg(f)$ to be the maximal size of a set $S \subseteq [n]$ such that $f_S \neq 0$.

We will need the following two simple observations:
\begin{itemize}
    \item $\deg(fg) \leq \deg(f) + \deg(g)$.
    \item If $g^{\leq D} = 0$ then $(fg)^{\leq D - \deg(f)} = 0$.
\end{itemize}
The first observation is simple. As for the second observation, let
\[
 f = \sum_{|S| \leq d} f_S, \quad g = \sum_{|T|>D} g_T.
\]
If we average $f_S g_T$ over any $i \in T \setminus S$ then it vanishes (choose $R = \overline{\{i\}} \not\supseteq T$). Since $|T \setminus S| > D-d$, this means that for every set $U$ of size at most $D-d$, if we average $fg$ over the coordinates outside $U$, then it vanishes. Induction on $|U|$ shows that $(fg)_U = 0$ for all $|U| \leq D-d$.

\smallskip

We use connectivity via a certain averaging operator.

\begin{definition}[Averaging operator] \label{def:averaging}
Let $\mathcal{D}$ be a connected distribution over $\{-1,1\}^k$, let $n \in \mathbb{N}$, and let $i \in [k-1]$. We consider $x \sim \mathcal{D}^n$, denoting the $k$ projections by $x^{(1)},\ldots,x^{(k)} \in \{-1,1\}^n$.

For $i \in [k-1]$, define an operator $S_i$ which takes a function $f$ depending on $x^{(1)},\ldots,x^{(i)}$ and outputs a function $S_if$ depending on $x^{(i+1)},\ldots,x^{(k)}$, defined by
\[
 (S_if)(y^{(i+1)},\ldots,y^{(k)}) = \E[f(x^{(1)},\ldots,x^{(i)}) \mid x^{(i+1)}=y^{(i+1)},\ldots,x^{(k)}=y^{(k)}].
\]

When $i = k$, we similarly define an operator $S_k$ which takes a function depending on $x^{(k)}$ and output a function $S_kf$ depending on $x^{(1)},\ldots,x^{(k-1)}$, defined in an analogous way.
\end{definition}

\begin{lemma} \label{lem:averaging}
Let $\mathcal{D}$ be a connected distribution on $\{-1,1\}^k$. There exists $\rho < 1$ such that the following holds for all $n \in \mathbb{N}$ and $i \in [k-1]$.

If $f$ is a function depending on the projection of $\mathcal{D}^n$ to the first $i$ coordinates and
$f^{\leq D} = 0$ then \[ \|S_i f\|_2 \leq \rho^D \|f\|_2. \]
\end{lemma}
\begin{proof}
Follows from \cite[Propositions 2.11 and 2.12]{Mossel2010}.
\end{proof}

We can now prove \Cref{thm:connected}, making free use of \Cref{lem:hypercontractivity}.

The proof uses a hybrid argument. For every $i \in [k]$, we will show that
\begin{equation} \label{eq:connected-hybrid}
 \bigl|\E_{x \sim \mathcal{D}^n}[|T_{1-\gamma_1} f_1(x^{(1)}) \cdots T_{1-\gamma_{i-1}} f_{i-1}(x^{(i-1)}) (1-T_{1-\gamma_i}) f_i(x^{(i)}) f_{i+1}(x^{(i+1)}) \cdots f_k(x^{(k)})]\bigr| \leq \frac{\epsilon}{k}.
\end{equation}

We will define the parameters $\gamma_1,\ldots,\gamma_{k-1}$ inductively, and $\gamma_k$ directly.

Suppose first that $i = k$. Since $T_{1-\gamma_1} f_1,\ldots,T_{1-\gamma_{k-1}} f_{k-1}$ are bounded, we can bound the \emph{square} of the left-hand side of~\eqref{eq:connected-hybrid} by
\[
 \E[|S_k(1-T_{1-\gamma_k}) f_k|]^2 \leq
 \E[(S_k(1-T_{1-\gamma_k}) f_k)^2] \leq
 \sum_{e \geq 0} \rho^{2e} (1-(1-\gamma_k)^e) \E[(f_k^{=e})^2] \leq \max_{e \geq 0} \rho^{2e} e\gamma_k,
\]
using $1 - (1-\gamma_k)^e \leq e\gamma_k$ and
$\E[f_k^2] = 1$. We can choose $\gamma_k > 0$ so that the maximum on the right-hand side is at most $(\epsilon/k)^2$.

Suppose next that $i < k$.
Since $f_{i+1},\ldots,f_k$ are bounded, we can bound the left-hand side of~\eqref{eq:connected-hybrid} by
\[
 \E[|S_i(T_{1-\gamma_1} f_1\ldots T_{1-\gamma_{i-1}} f_{i-1} (1-T_{1-\gamma_i}) f_i)|].
\]
Note that $T_{1-\gamma_j} f_j$ is bounded in magnitude by $1$ for all $j < i$, and $(1-T_{1-\gamma_i}) f_i$ is bounded in magnitude by $2$.

We will inductively find parameters $d_1,\ldots,d_i \in \mathbb{N}$ such that the following three inequalities hold.
First, for all $j \in [i-1]$,
\begin{equation} \label{eq:connected-1}
 \E[|T_{1-\gamma_1} f_1^{\leq d_1} \cdots T_{1-\gamma_{j-1}} f_{j-1}^{\leq d_{j-1}} T_{1-\gamma_j} f_j^{>d_j}|] \leq \frac{\epsilon}{2k(i+1)}.
\end{equation}
Second,
\begin{equation} \label{eq:connected-2}
 \E[|T_{1-\gamma_1} f_1^{\leq d_1} \cdots T_{1-\gamma_{i-1}} f_{i-1}^{\leq d_{i-1}} (1-T_{1-\gamma_i}) f_i^{\leq d_i}|] \leq \frac{\epsilon}{k(i+1)}.
\end{equation}
Third,
\begin{equation} \label{eq:connected-3}
 \E[|S_i (T_{1-\gamma_1} f_1^{\leq d_1} \cdots T_{1-\gamma_{i-1}} f_{i-1}^{\leq d_{i-1}} (1-T_{1-\gamma_i}) f_i^{> d_i})|] \leq \frac{\epsilon}{k(i+1)}.
\end{equation}
Together, these imply \eqref{eq:connected-hybrid}, since $S_i$ is contractive. %It will be convenient to work with $d_{<j} = d_1 + \cdots + d_{j-1}$

Let us start with~\eqref{eq:connected-1}. Generalized H\"older's inequality shows that the left-hand side of~\eqref{eq:connected-1} is bounded by
\[
 \prod_{\ell=1}^{j-1} \E[|T_{1-\gamma_\ell} f_\ell^{\leq d_\ell}|^{2(j-1)}]^{1/2(j-1)} \E[(T_{1-\gamma_j} f_j^{>d_j})^2]^{1/2} \leq
 (M\sqrt{2j-3})^{d_1 + \cdots + d_{j-1}} (1-\gamma_j)^{d_j}.
\]
Given $d_1,\ldots,d_{j-1},\gamma_j$, we can choose $d_j$ so that \eqref{eq:connected-1} is satisfied.

The argument for~\eqref{eq:connected-2} is similar. Generalized H\"older's inequality bounds the left-hand side of~\eqref{eq:connected-2} by
\[
 \prod_{\ell=1}^{i-1} \E[|T_{1-\gamma_\ell} f_\ell^{\leq d_\ell}|^{2(i-1)}]^{1/2(i-1)} \E[((1-T_{1-\gamma_i}) f_j^{\leq d_i})^2]^{1/2} \leq
 (M\sqrt{2i-3})^{d_1 + \cdots + d_{i-1}} d_i \gamma_i.
\]
Given $d_1,\ldots,d_i$, we can choose $\gamma_i$ so that \eqref{eq:connected-2} holds. We choose $d_i$ below, \emph{independently of $\gamma_i$}.

Finally, we tackle~\eqref{eq:connected-3}. Applying Cauchy--Schwarz, it suffices to show that
\[
 \E[(S_i (T_{1-\gamma_1} f_1^{\leq d_1} \cdots T_{1-\gamma_{i-1}} f_{i-1}^{\leq d_{i-1}} (1-T_{1-\gamma_i}) f_i^{> d_i}))^2]^{1/2} \leq \frac{\epsilon}{k(i+1)}.
\]
Let $d_{<i} = d_1 + \cdots + d_{i-1}$. The function $F$ that $S_i$ is applied to satisfies $F^{\leq d_i - d_{<i}} = 0$, and so we can use \Cref{lem:averaging} to bound the left-hand side by
\begin{multline*}
 \rho^{d_i - d_{<i}}
  \E[(T_{1-\gamma_1} f_1^{\leq d_1} \cdots T_{1-\gamma_{i-1}} f_{i-1}^{\leq d_{i-1}} (1-T_{1-\gamma_i}) f_i^{> d_i})^2]^{1/2} \leq \\
  \rho^{d_i - d_{<i}} \E[(T_{1-\gamma_1} f_1^{\leq d_1} \cdots T_{1-\gamma_{i-1}} f_{i-1}^{\leq d_{i-1}} (1-T_{1-\gamma_i}) f_i)^2]^{1/2} + \\
  \rho^{d_i - d_{<i}} \E[(T_{1-\gamma_1} f_1^{\leq d_1} \cdots T_{1-\gamma_{i-1}} f_{i-1}^{\leq d_{i-1}} (1-T_{1-\gamma_i}) f_i^{\leq d_i})^2]^{1/2},
\end{multline*}
using the triangle inequality. We will show that both summands are bounded by $\frac{\epsilon}{2k(i+1)}$. We present the proof for the second summand; the same argument also bounds the first summand.

Applying the generalized H\"older inequality, for any $q_1,\ldots,q_i \geq 2$ such that $1/q_1 + \cdots + 1/q_i = 1/2$ we can bound the second summand by
\[
 \rho^{d_i - d_{<i}}\prod_{j=1}^{i-1} \E[|T_{1-\gamma_j} f_j^{\leq d_j}|^{q_j}]^{1/q_j}
 \E[|(1-T_{1-\gamma_i}) f_i^{\leq d_i}|^{q_i}]^{1/q_i}
 \leq
 2\rho^{d_i - d_{<i}}\prod_{j=1}^i (M^{1/2-1/q_j} \sqrt{q_j-1})^{d_j},
\]
since $T_{1-\gamma_1} f_1,\ldots,T_{1-\gamma_{i-1}}f_{i-1}$ are bounded by $1$ and $(1-T_{1-\gamma_i}) f_i$ is bounded by $2$.

We choose $q_1,\ldots,q_i$ as follows. Let $q$ be a large parameter. We take $q_1 = \cdots = q_{i-1} = 2(i-1)q$ and $q_i = \frac{2q}{q-1}$. We can bound the summand by
\[
 2\rho^{d_i - d_{<i}} \sqrt{2Miq}^{d_{<i}} \left(M^{1/2q} \sqrt{1+\frac{2}{q-1}}\right)^{d_i}.
\]
Choose $q$ large enough so that $\rho \cdot M^{1/2q} \sqrt{1+\frac{2}{q-1}} \leq \sqrt{\rho}$. We can bound the summand by
\[
 2\sqrt{\rho}^{d_i} (\sqrt{2Miq}/\rho)^{d_{<i}}.
\]
We can choose $d_i$ to make this at most $\frac{\epsilon}{2k(i+1)}$, completing the proof.

\section{Approximate polymorphisms using decision trees}
\label{apx:decision-trees}

\Cref{thm:jones-dt} is a version of \Cref{thm:jones} for decision trees. The main advantage of \Cref{thm:jones-dt} over \Cref{thm:jones} is the vastly improved parameters.

We can prove \Cref{lem:f-junta-general} using \Cref{thm:jones-dt} rather than \Cref{thm:jones}; thus $T$ is now a decision tree rather than a junta. The main difference is in the way in which $Z,W$ are chosen:

\begin{itemize}
    \item Choose $Z$ uniformly at random.
    \item Let $B$ be the set of variables that $T$ encounters when reading $\col_m(Z)$.
    \item Let $R \subseteq \overline{B}$ consist of those rows $i \notin B$ such that $\row_i(Z)|_{[m-1]} = \alpha$.
    \item Define $W$ to be the matrix obtained from $Z$ by resampling $W_{im}$ for $i \in R$.
\end{itemize}
The rest of the proof of \Cref{lem:f-junta-general} goes through without changes.

We can obtain \Cref{lem:f-junta-general} as stated by replacing the resulting decision tree of depth $L$ by a junta depending on all variables mentioned by the decision tree; there are fewer than $2^{L+1}$ of these.

To see what parameters we obtain, notice first that \Cref{thm:it-aint-over} holds with $\delta = \epsilon^{O(1/\sqrt{\log(1/\epsilon)})}$, $\tau = \epsilon^{O(\sqrt{\log(1/\epsilon)})},$ and $d = O(\log^{3/2}(1/\epsilon))$, where $\epsilon = \min(\epsilon_1,\epsilon_2) = \epsilon^{\Theta(1)}$. We therefore get
\[
 L = O(d/\tau \epsilon) = (1/\epsilon)^{O(1/\sqrt{\log(1/\epsilon)})}
\]
and
\[
 \eta = \epsilon^{O(1/\sqrt{\log(1/\epsilon)})}.
\]

In the proof of \Cref{lem:approximate-polymorphisms-not-xor} we choose $\delta = \min(\epsilon,\eta,2^{-m2^{L+1}}/3)$ (our junta is on $2^{L+1}$ variables), and so we get
\[
 1/\delta = 2^{2^{(1/\epsilon)^{O(1/\sqrt{\log(1/\epsilon)})}}}.
\]

In contrast, the argument in~\cite{FLMM2020} gives $1/\delta = 2^{(1/\epsilon)^{O(1)}}$, as worked out in~\cite{SimpleAND}.

\section{Approximate polymorphisms of \texorpdfstring{$\Maj_3$}{Maj3}} \label{apx:Maj3}

In this section, we prove the following result.

\begin{lemma} \label{lem:maj3}
Any $\epsilon$-approximate polymorphism of $\Maj_3$, the majority function on three inputs, is $O(\epsilon)$-close to a constant, a dictator, or an anti-dictator.
\end{lemma}
\begin{proof}
Let $F\colon \{-1,1\}^n \to \{-1,1\}$ be given by $F((-1)^x) = (-1)^{f(x)}$. The fact that $f$ is an $\epsilon$-approximate polymorphism of $\Maj_3$ translates to the following fact about $F$:
\[
 \E_{x,y,z \in \{-1,1\}^n}[|F(\Maj_3(x,y,z)) - \Maj_3(F(x),F(y),F(z))|] \leq 2\epsilon,
\]
and so $|\E[F] - \Maj_3(\E[F],\E[F],\E[F])| \leq 2\epsilon$. Since $\Maj_3(\E[F],\E[F],\E[F]) = \E[F] + \frac{1-\E[F]^2}{2} \E[F]$, it follows that $\E[F]$ is $O(\epsilon)$-close to $\{-1,0,1\}$. If $\E[F]$ is $O(\epsilon)$-close to $\{-1,1\}$ then $f$ is $O(\epsilon)$-close to constant. Otherwise, since
\[
 \E_{x,y,z \in \{-1,1\}^n}[(F(\Maj_3(x,y,z)) - \Maj_3(F(x),F(y),F(z)))^2] \leq 4\epsilon,
\]
considering the Fourier coefficients corresponding to $\prod_{i \in S} x_i$ for all $S \neq \emptyset$, we have
\[
 \sum_{S \neq \emptyset} \frac{1-\E[F]^2}{2} \hat{F}(S)^2 \leq 4\epsilon.
\]
Since $\E[F]$ is close to~$0$, this shows that $\|F^{>1}\|^2 = O(\epsilon)$, and so the FKN theorem~\cite{FKN} implies that $f$ is close to a dictator or anti-dictator.
\end{proof}

\end{document}